\documentclass[a4paper,11pt]{easychair}

\usepackage[makeroom]{cancel} 
\usepackage{hyperref}
\usepackage{amsmath}
\usepackage{amsthm}
\usepackage{amssymb}
\usepackage{mathrsfs}
\usepackage{graphicx}
\usepackage{stfloats}
\usepackage[normalem]{ulem} 
\usepackage[pdftex]{xcolor}
\usepackage{paralist}
\usepackage{enumerate}
\usepackage{url}
\usepackage{microtype}
\usepackage[all]{xy}
\def\romannum{\begingroup
  \def\theenumi{\textup{(\roman{enumi})}}%
  \def\p@enumi{}%
  \def\labelenumi{\theenumi}%
  \enumerate}

  {\end{enumerate}}%
  {\end{enumerate}}

\newcommand{\stabilo}{\bgroup\markoverwith
  {\textcolor{yellow}{\rule[-.5ex]{2pt}{2.5ex}}}\ULon}
\newcommand{\newstuff}{\bgroup\markoverwith{\textcolor{pink}{\rule[-.5ex]{2pt}{2.5ex}}}\ULon}

\newtheorem{theorem}{Theorem}

\newtheorem{corollary}[theorem]{Corollary}
\newtheorem{application}[theorem]{Application}
\newtheorem{conjecture}[theorem]{Conjecture}
\newtheorem{fact}[theorem]{Fact}
\newtheorem{lemma}[theorem]{Lemma}
\newtheorem{proposition}[theorem]{Proposition}

\theoremstyle{example}
\newtheorem{remark}[theorem]{Remark}
\newtheorem{example}[theorem]{Example}
\newtheorem{question}[theorem]{Question}

\newcommand{\notmodels}{\ \makebox[0.1cm][l]{\ensuremath{\models}}/ \ }

\renewcommand{\phi}{\varphi}

\newcommand{\ignore}[1]{}

\newcommand{\integerset}[0]{\ensuremath{\mathbb{N}}}
\newcommand{\QCSP}[0]{\ensuremath{\textrm{QCSP}}}
\newcommand{\CSP}[0]{\ensuremath{\textrm{CSP}}}
\newcommand{\NP}[0]{\ensuremath{\textrm{NP}}}
\newcommand{\coNP}[0]{\ensuremath{\textrm{co\mbox{-}NP}}}
\newcommand{\Pspace}[0]{\ensuremath{\textrm{Pspace}}}
\newcommand{\Inv}[0]{\ensuremath{\textrm{Inv}}}
\newcommand{\Pol}[0]{\ensuremath{\textrm{Pol}}}
\newcommand{\restrict}{\ensuremath{\upharpoonright}}

\newcommand{\reactivelycomposable}{\ensuremath{\trianglelefteq}}

\title{The complexity of quantified constraints: collapsibility, switchability and the algebraic formulation}
\titlerunning{The complexity of quantified constraints}
\author{Catarina Carvalho\inst{1}
  \and
  Florent R. Madelaine\inst{2}
  \and
  Barnaby Martin\inst{3}
 \and
 Dmitriy Zhuk\inst{4}
}

\institute{
  School of Physics, Astronomy and Mathematics,
  University of Hertfordshire, UK.
  \email{c.carvalho2@herts.ac.uk}
  \and 
  LACL, Universit\'e Paris-Est Cr\'eteil, France.
  \email{florent.madelaine@gmail.com}
  \and
  Department of Computer Science,
  Durham University, UK.
  \email{barnabymartin@gmail.com}
  \and
Department of Mechanics and Mathematics,
Lomonosov Moscow State University, Russia.
  \email{zhuk.dmitriy@gmail.com}
}
\authorrunning{Carvalho, Madelaine, Martin and Zhuk}

\begin{document}
\maketitle

\begin{abstract}
Let $\mathbb{A}$ be an idempotent algebra on a finite domain. By mediating between results of Chen \cite{AU-Chen-PGP} and Zhuk \cite{ZhukGap2015}, we argue that if $\mathbb{A}$ satisfies the polynomially generated powers property (PGP) and $\mathcal{B}$ is a constraint language invariant under $\mathbb{A}$ (that is, in $\mathrm{Inv}(\mathbb{A})$), then QCSP$(\mathcal{B})$ is in NP. In doing this we study the special forms of PGP, switchability and collapsibility, in detail, both algebraically and logically, addressing various questions such as decidability on the way.

We then prove a complexity-theoretic converse in the case of infinite constraint languages encoded in propositional logic, that if $\mathrm{Inv}(\mathbb{A})$ satisfies the exponentially generated powers property (EGP), then QCSP$(\mathrm{Inv}(\mathbb{A}))$ is co-NP-hard. Since Zhuk proved that only PGP and EGP are possible, we derive a full dichotomy for the QCSP, justifying what we term the Revised Chen Conjecture. This result becomes more significant now the original Chen Conjecture (see \cite{Meditations}) is known to be false \cite{ZhukM20}.

Switchability was introduced by Chen in  \cite{AU-Chen-PGP} as a generalisation of the already-known collapsibility \cite{hubie-sicomp}. For three-element domain algebras $\mathbb{A}$ that are switchable and omit a G-set, we prove that, for every finite subset $\Delta$ of Inv$(\mathbb{A})$, Pol$(\Delta)$ is collapsible. The significance of this is that, for QCSP on finite structures (over a three-element domain), all QCSP tractability (in P)  explained by switchability is already explained by collapsibility.

\end{abstract}

\section{Introduction}

A large body of work exists from the past twenty years on
applications of universal algebra to the computational complexity of
\emph{constraint satisfaction problems} (CSPs) and a number of
celebrated results have been obtained through this method. 
One considers the problem CSP$(\mathcal{B})$ in which it is asked whether an input
sentence $\varphi$ holds on $\mathcal{B}$, a \emph{constraint language} (equivalently, relational structure), where $\varphi$ is
\emph{primitive positive}, that is using only $\exists$, $\land$ and $=$. The CSP is one of a wide class of model-checking problems obtained from restrictions of first-order logic. For almost all of these classes, we can give a complexity
classification~\cite{DBLP:journals/corr/abs-1210-6893}. Since the celebrated proofs of the Feder-Vardi ``Dichotomy'' Conjecture for CSPs \cite{BulatovFVConjecture,ZhukFVConjecture,Zhuk20}, the only
outstanding class (other than its natural dual) is \textsl{quantified
  CSPs} (QCSPs) for positive Horn sentences -- where $\forall$ is also present -- which is used in
Artificial Intelligence to model non-monotone reasoning or
uncertainty \cite{nonmonotonic}. 

It is well-known that the complexity classification for QCSPs embeds the classification for CSPs: if $\mathcal{B}+1$ is $\mathcal{B}$ with the addition of a new isolated element not appearing in any relations, then CSP$(\mathcal{B})$ and QCSP$(\mathcal{B}+1)$ are polynomially equivalent. Thus the classification for QCSPs may be considered a project at least as hard as that for CSPs. The following is the merger of Conjectures 6 and 7 in \cite{Meditations} which we call the \emph{Chen Conjecture}.
\begin{conjecture}[Chen Conjecture]
Let $\mathcal{B}$ be a finite relational structure expanded with constants naming all the elements. If $\Pol(\mathcal{B})$ has PGP, then $\QCSP(\mathcal{B})$ is in $\NP$; otherwise $\QCSP(\mathcal{B})$ is $\Pspace$-complete.
\end{conjecture}
\noindent In \cite{Meditations}, Conjecture 6 gives the NP membership and Conjecture 7 the Pspace-completeness. The first contribution of this paper is to prove that the NP membership of Conjecture 6 is indeed true. We do this by proving equivalent two notions of switchability that allows to combine known results from  \cite{AU-Chen-PGP} and \cite{ZhukGap2015}. On the way we develop the notions of non-degenerate and projective adversaries that enable us to prove our result as well as particular observations on the existing notions of switchability and collapsibility. 

The second contribution of this paper is Theorem~\ref{thm:all} below, but note that we permit infinite signatures (languages) although our domains remain finite. This will involve deciding how to encode relations of $\Inv(\mathbb{A})$ and will be discussed in detail later. 
\begin{theorem}[Revised Chen Conjecture]
Let $\mathbb{A}$ be an idempotent algebra on a finite domain $A$. If $\mathbb{A}$ satisfies PGP, then \QCSP$(\mathrm{Inv}(\mathbb{A}))$ is in \NP. Otherwise, \QCSP$(\mathrm{Inv}(\mathbb{A}))$ is \coNP-hard.
\label{thm:all}
\end{theorem}
\noindent Note that, with infinite languages, the NP-membership for Theorem~\ref{thm:all} requires a little extra work. We are also able to refute the following form.
\begin{conjecture}[Alternative Chen Conjecture]
\label{thm:alternative}
Let $\mathbb{A}$ be an idempotent algebra on a finite domain $A$. If $\mathbb{A}$ satisfies PGP, then for every finite subset  $\Delta \subseteq \mathrm{Inv}(\mathbb{A})$, \QCSP$(\Delta)$ is in \NP. Otherwise, there exists a finite subset $\Delta \subseteq \mathrm{Inv}(\mathbb{A})$ so that \QCSP$(\Delta)$ is \coNP-hard.
\end{conjecture}
\noindent The Alternative Chen Conjecture was not posed by Chen himself, but is nonetheless natural. In proving Theorem~\ref{thm:all} we are saying that the complexity of QCSPs, with all constants included, is classified modulo the complexity of (infinite language) CSPs, a subject to which we will return later.
\begin{corollary}
Let $\mathbb{A}$ be an idempotent algebra on a finite domain $A$. Either \QCSP$(\mathrm{Inv}(\mathbb{A}))$ is \coNP-hard or  \QCSP$(\mathrm{Inv}(\mathbb{A}))$ has the same complexity as \CSP$(\mathrm{Inv}(\mathbb{A}))$.
\end{corollary}
\noindent In this manner, our result follows in the footsteps of the similar result for the Valued CSP, which has also had its complexity classified modulo the CSP, as culminated in the paper \cite{FOCS2015}.

For a finite-domain algebra $\mathbb{A}$ we associate a function
$f_\mathbb{A}:\mathbb{N}\rightarrow\mathbb{N}$, giving the cardinality
of the minimal generating sets of the sequence $\mathbb{A},
\mathbb{A}^2, \mathbb{A}^3, \ldots$ as $f_\mathbb{A}(1), f_\mathbb{A}(2), f_\mathbb{A}(3), \ldots$,
respectively. A subset $\Lambda$ of $A^m$ is a generating set for $\mathbb{A}^m$ exactly if, for every $(a_1,\ldots,a_m) \in A^m$, there exists a $k$-ary term operation $f$ of $\mathbb{A}$ and $(b^1_1,\ldots,b^1_m),\ldots,$ $(b^k_1,\ldots,b^k_m) \in \Lambda$ so that $f(b^1_1,\ldots,b^k_1)=a_1$, \ldots, $f(b^1_m,\ldots,b^k_m)=a_m$. We may say $\mathbb{A}$ has the $g$-GP if $f_\mathbb{A}(m) \leq
g(m)$ for all $m$. The question then arises as to the growth rate of
$f_\mathbb{A}$ and specifically regarding the behaviours constant, logarithmic,
linear, polynomial and exponential. Wiegold proved in
\cite{WiegoldSemigroups} that if $\mathbb{A}$ is a finite semigroup
then $f_{\mathbb{A}}$ is either linear or exponential, with the former
prevailing precisely when $\mathbb{A}$ is a monoid. This dichotomy
classification may be seen as a gap theorem because no growth rates
intermediate between linear and exponential may occur. We say
$\mathbb{A}$  enjoys the \emph{polynomially generated powers} property
(PGP) if there exists a polynomial $p$ so that $f_{\mathbb{A}}=O(p)$
and  the \emph{exponentially generated powers} property (EGP) if there
exists a constant $b$ so that $f_{\mathbb{A}}=\Omega(g)$ where
$g(i)=b^i$.

In Chen's \cite{AU-Chen-PGP}, a new link between algebra and
QCSP was discovered. Chen's previous work in QCSP tractability largely
involved the special notion of \emph{collapsibility}
\cite{hubie-sicomp}, but in \cite{AU-Chen-PGP} this was extended to a \emph{computationally effective} version of the PGP. For a finite-domain, idempotent algebra $\mathbb{A}$, call \emph{simple $k$-collapsibility}\footnote{We want to use a name different from ``collapsibility'' alone in order to differentiate this from Chen's original definition. In \cite{MFCS2017} we used capitalisation, with a leading capital letter for Chen's original version and all small letters for what we here designate simple.} that special form of the PGP in which the generating set for $\mathbb{A}^m$ is constituted of all tuples $(x_1,\ldots,x_m)$ in which at least $m-k$ of these elements are equal. \emph{Simple $k$-switchability} will be another special form of the PGP in which the generating set for $\mathbb{A}^m$ is constituted of all tuples $(x_1,\ldots,x_m)$ in which there exists $a_i<\ldots<a_{k'}$, for $k'\leq k$, so that
\[ (x_1,\ldots,x_m) = (x_1,\ldots,x_{a_1},x_{a_1+1},\ldots,x_{a_2},x_{a_2+1},\ldots,\ldots,x_{a_k'},x_{a_k'+1},\ldots,x_m),\]
where $x_1=\ldots=x_{a_1-1}$, $x_{a_1}=\ldots=x_{a_2-1}$, \ldots, $x_{a_{k'}}=\ldots=x_{a_m}$. Thus, $a_1,a_2,\ldots,a_{k'}$ are the indices where the tuple switches value.  We say that $\mathbb{A}$ is simply collapsible (switchable) if there exists $k$ such that it is simply $k$-collapsible ($k$-switchable). We note that Zhuk uses this form of simple switchability,  in \cite{ZhukGap2015}, where he proves that the only kind of PGP for finite-domain algebras is simple switchability.

Our first contribution shows $k$-collapsibility, whose definition is deferred until adversaries are introduced in Section \ref{sec:from-pgp-complexity}, and simple $k$-collapsibility, coincide. The same applies to $k$-switchability and simple $k$-switchability, and we will dwell on these distinctions no longer. For any finite algebra, $k$-collapsibility implies $k$-switchability, and for any $2$-element algebra, $k$-switchability implies $k$-collapsibility.

Switchability was introduced by Chen in  \cite{AU-Chen-PGP} as a generalisation of the already-known collapsibility \cite{hubie-sicomp} when he discovered a $4$-ary operation $f$ on the three-element domain so that $\{f\}$ had the PGP (switchability) but was not collapsible. Thus it seemed that collapsibility was not enough to explain membership of QCSP in NP. What we prove as our third contribution is that Inv$\{f\}$ is not finitely related, and what is more, every finite subset of Inv$\{f\}$ is collapsible. Moreover, we prove this for all switchable clones $\mathbb{A}$ on $3$-elements that omit a G-set, what Chen terms \emph{Gap Algebras}.
For QCSP complexity for three-element structures, our result already shows we do not need the additional notion of switchability to explain membership in P, as collapsibility will already suffice. Note that the parameter $k$ of collapsibility is unbounded over these increasing finite subsets while the parameter of switchability clearly remains bounded. 

In the arxiv version of \cite{MFCS2017} we proved that this is also true when $\mathbb{A}$ has a G-set. The proof is an exhaustive case analysis and is neither interesting nor sheds light on the general case. It is omitted due to its longevity and the fact that Chen was most interested in the Gap Algebras (\mbox{cf.} the ``Classification'' Theorem 8.1 in \cite{AU-Chen-PGP}).  If these results were generalisable to higher domains then perhaps collapsibility is enough to explain all membership of finite constraint language QCSP in NP.

\subsection{Infinite languages}

Our use of infinite languages (\mbox{i.e.} signatures, since we work on a finite domain) is a controversial part of our discourse and merits special discussion. We wish to argue that a necessary corollary of the algebraic approach to (Q)CSP is a reconciliation with infinite languages. The traditional approach to consider arbitrary finite subsets of Inv$(\mathbb{A})$ is unsatisfactory in the sense that choosing this way to escape the -- naturally infinite -- set Inv$(\mathbb{A})$ is as arbitrary as the choice of encoding required for infinite languages. However, the difficulty in that choice is of course the reason why this route is often eschewed. The first possibility that comes to mind for encoding a relation in Inv$(\mathbb{A})$ is probably to list its tuples, while the second is likely to be to describe the relation in some kind of ``simple'' logic. Both these possibilities are discussed in \cite{Creignou}, for the Boolean domain, where the ``simple'' logic is the propositional calculus. For larger domains, this would be equivalent to quantifier-free propositions over equality with constants. Both Conjunctive Normal Form (CNF) and Disjunctive Normal Form (DNF) representations are considered in \cite{Creignou} and a similar discussion in \cite{ecsps} exposes the advantages of the DNF encoding. The point here is that testing non-emptiness of a relation encoded in CNF may already be NP-hard, while for DNF this will be tractable. Since DNF has some benign properties, we might consider it a ``nice, simple'' logic while for ``simple'' logic we encompass all quantifier-free sentences, that include DNF and CNF as special cases. The reason we describe this as ``simple'' logic is to compare against something stronger, say all first-order sentences over equality with constants. Here recognising non-emptiness becomes Pspace-hard and since QCSPs already sit in Pspace, this complexity is unreasonable.

For the QCSP over infinite languages Inv$(\mathbb{A})$, Chen and Mayr \cite{QCSPmonoids} have declared for our first, tuple-listing, encoding. In this paper we will choose the ``simple'' logic encoding, occasionally giving more refined results for its ``nice, simple'' restriction to DNF. Our choice of the ``simple'' logic encoding over the tuple-listing encoding will ultimately be justified by the (Revised) Chen Conjecture holding for ``simple'' logic yet failing for tuple-listings. Since the original Chen Conjecture is known now to be false \cite{ZhukM20}, our result becomes more remarkable. However, there are some surprising consequences, it follows from \cite{ZhukM20} that there exists a finite and 3-element $\mathcal{B}$ with constants, so that QCSP$(\mathrm{Inv}(\mathrm{Pol}(\mathcal{B})))$, under our encoding, and QCSP$(\mathcal{B})$ have different complexities: the former being co-NP-hard while the latter is in P.

The Feder-Vardi Conjecture is known to hold for infinite languages \cite{ZhukFVConjecture-infinite} but the proofs are based on the tuple-listing encoding. We can not say whether the polynomial cases are preserved under the DNF encoding.

Let us consider examples of our encodings. For the domain $\{1,2,3\}$, we may give a binary relation either by the tuples $\{ (1,2), (2,1), (2,3), (3,2), (1,3), (3,1), (1,1) \}$ or by the ``simple'' logic formula $(x\neq y \vee x=1)$. For the domain $\{0,1\}$, we may give the ternary (not-all-equal) relation by the tuples $\{ (1,0,0), (0,1,0), (0,0,1), (1,1,0), (1,0,1), (1,1,0)\}$ or by the ``simple'' logic formula $(x\neq y \vee y \neq z)$. In both of these examples, the simple formula is also in DNF.

\vspace{0.2cm}
\noindent \textbf{Nota Bene}. The results of this paper apply for the ``simple'' logic encoding as well as the ``nice, simple'' encoding in DNF except where specifically stated otherwise. These exceptions are Proposition~\ref{prop:coNP2} and Corollary~\ref{cor:coNP2} (which  uses the ``nice, simple'' DNF) and Proposition~\ref{prop:Chen-fails} (which uses the tuple-listing encoding).

\subsection{Related work}

This is the journal version of \cite{LICS2015} and \cite{MFCS2017}. The majority of the proofs were omitted from these conference papers but the section numbers are preserved in the arxiv versions. However, several parts of those papers have become superseded or otherwise outdated. This applies to Sections 3 and 5 of  \cite{LICS2015}, leaving Section 4 appearing in its entirety. From \cite{MFCS2017} we give Section 3 in its entirety but only the more interesting part of Section 4 (omitting the algebras containing a G-set). Section 5 is omitted.

On the other hand, the canonical example of projective and non-degenerate adversaries is now known to be switchability \cite{ZhukGap2015}. This has raised the importance of Section 4 of  \cite{LICS2015} as the bridge between two forms of switchability and a necessary part of proving that PGP yields a QCSP in NP.

\subsection{Some comment on notation}

We us calligraphic notation $\mathcal{A}$ for constraint languages over domain $A$. Constraint languages can be seen as a set of relations over the same domain or as first-order relational structures and we rather conflate the two (already in the abstract). Sets such as $\Inv(\mathbb{A})$ can be seen as infinite constraint languages and we might talk of (finite) subsets $\Delta$ of this as a constraint language or a (finite-signature) reduct. 

Algebras are indicates in blackboard notation $\mathbb{A}$. All domains in this paper are finite. We write pH to indicate positive Horn.  

\newcommand{\BarnyApprovedTitleForOurStuffIntro}[0]{The PGP: collapsibility and beyond}
\newcommand{\BarnyApprovedTitleForOurStuffText}[0]{\BarnyApprovedTitleForOurStuffIntro}

\section{\BarnyApprovedTitleForOurStuffText}
\label{sec:from-pgp-complexity}
Throughout this section, we will be concerned with a constraint language $\mathcal{A}$ that may or may not have some constants naming the elements. We will be specific when we require constants naming elements. In Chen's \cite{hubie-sicomp,AU-Chen-PGP}, the assumption of constants naming elements is often implicit, \mbox{e.g.} through idempotency, but several of his theorems apply in the general case, and are reproduced here in generality. 

Later in this section we will use Fraktur notation for constraint languages embellished with additional constants (different from any basic constants just naming elements) that we ultimately use to denote universal variables.

\subsection{Games, adversaries and reactive composition}
\label{sec:games-adversaries-reactivecomposition}
%
We recall some terminology due to
Chen~\cite{hubie-sicomp,AU-Chen-PGP}, for his natural adaptation of
the model checking game to the context of pH-sentences. 
We shall not need to explicitly play these games but only to handle
strategies for the existential player. 
An \emph{adversary} $\mathscr{B}$ of length $m\geq 1$ is an $m$-ary
relation over $A$. 
When $\mathscr{B}$ is precisely the set
$B_{1}\times B_{2} \times \ldots \times B_{m}$ for some non-empty
subsets $B_1,B_2,\ldots,B_m$ of $A$, we speak of a \emph{rectangular
  adversary}.
Let $\phi$ have universal variables $x_1,\ldots,x_m$ and quantifier-free
part $\psi$. We write $\mathcal{A}\models \phi_{\restrict\mathscr{B}}$ and say that
\emph{the existential player has a winning strategy in the
$(\mathcal{A},\phi)$-game against adversary $\mathscr{B}$} iff there
exists a set of Skolem functions $\{\sigma_x : \mbox{`$\exists x$'}
\in \phi \}$ such that for any assignment $\pi$ of the universally
quantified variables of $\varphi$ to $A$, where
$\bigl(\pi(x_1),\ldots,\pi(x_m)\bigr) \in \mathscr{B}$, the map $h_\pi$ 
is a homomorphism from $\mathcal{D}_\psi$ (the canonical database) to $\mathcal{A}$, where
$$h_\pi(x):=
\begin{cases}
  \pi(x) & \text{, if $x$ is a universal variable; and,}\\
  \sigma_x(\left.\pi\right|_{Y_x})& \text{, otherwise.}\\
\end{cases}
$$
(Here, $Y_x$ denotes the set of universal variables preceding $x$ and $\left.\pi\right|_{Y_x}$ the restriction of $\pi$ to $Y_x$.)
Clearly, $\mathcal{A} \models \phi$ iff the existential player has a winning strategy in the
$(\mathcal{A},\phi)$-game against the so-called \emph{full
  (rectangular) 
  adversary} $A\times A \times \ldots \times A$ (which
we will denote hereafter by $A^m$).
We say that an adversary $\mathscr{B}$ of length $m$ \emph{dominates}
an adversary $\mathscr{B}'$ of length $m$ when $\mathscr{B}'\subseteq
\mathscr{B}$. Note that $\mathscr{B}'\subseteq  \mathscr{B}$ and
$\mathcal{A}\models \phi_{\restrict\mathscr{B}}$ implies
$\mathcal{A}\models \phi_{\restrict\mathscr{B}'}$. 
We will also
consider sets of adversaries of the same length, denoted
by uppercase Greek letters as in $\Omega_m$; and, sequences thereof, which we
denote with bold uppercase Greek letters as in
$\mathbf{\Omega}=\bigl(\Omega_m\bigr)_{m \in \integerset}$. 
We will write
$\mathcal{A}\models \phi_{\restrict\Omega_m}$ to denote that
$\mathcal{A}\models  \phi_{\restrict\mathscr{B}}$ holds for every
adversary $\mathscr{B}$ in $\Omega_m$. 
We call \emph{width} of $\Omega_m$ and write $\mathrm{width}(\Omega_m)$ for $\sum_{\mathscr{B}\in
  \Omega_m} |\mathscr{B}|$.
We say that $\mathbf{\Omega}$ is \emph{polynomially bounded} if
there exists a polynomial $p(m)$ such that for every $m\geq 1$,
$\mathrm{width}(\Omega_m) \leq p(m)$. We say that $\mathbf{\Omega}$ is
\emph{effective} if there exists a polynomial $p'(m)$ and an algorithm that outputs
$\Omega_m$ for every $m$ in total time
$p'(\mathrm{width}(\Omega_m))$.
%
%

Let $f$ be a $k$-ary operation of $\mathcal{A}$ and $\mathscr{A},\mathscr{B}_1,\ldots,\mathscr{B}_k$ be adversaries of length $m$.
We say that $\mathscr{A}$ is \emph{reactively composable} from the
adversaries $\mathscr{B}_1,\ldots,\mathscr{B}_k$ via $f$, and we write $\mathscr{A} \reactivelycomposable f(\mathscr{B}_1,\ldots,\mathscr{B}_k)$ iff there exist partial functions $g^j_i:A^i \to A$ for every $i$ in $[m]$ and every $j$ in $[k]$ such that, for every tuple $(a_1,\ldots,a_m)$ in adversary $\mathscr{A}$ the following holds.
\begin{compactitem}
\item for every $j$ in $[k]$, the values 
  $g^j_1(a_1), g^j_2(a_1,a_2),$ $\ldots , g^j_m(a_1,a_2,\ldots,a_m)$ are defined and the tuple $\bigl(g^j_1(a_1), g^j_2(a_1,a_2), \ldots, g^j_m(a_1,a_2,\ldots,a_m)\bigr)$ is in adversary $\mathscr{B}_j$; and,
\item for every $i$ in $[m]$, $a_i =f\bigl(g^1_i(a_1,a_2,\ldots,a_i),$
  $g^2_i(a_1,a_2,\ldots,a_i),\ldots,g^k_i(a_1,a_2,\ldots,a_i))$.
\end{compactitem}
We write $\mathscr{A} \reactivelycomposable
\{\mathscr{B}_1,\ldots,\mathscr{B}_k\}$ if there exists a $k$-ary
operation $f$ such that $\mathscr{A} \reactivelycomposable f(\mathscr{B}_1,\ldots,\mathscr{B}_k)$

\begin{remark}
  We will never show reactive composition by exhibiting a polymorphism $f$ and  partial functions $g^i_j$ that depend on all their arguments. We will always be able to exhibit partial functions that depend only on their last argument. 
\end{remark}
Reactive composition allows to interpolate complete Skolem functions
from partial ones.
\begin{theorem}[{\cite[Theorem 7.6]{AU-Chen-PGP}}]
  \label{thm:hubieReactiveComposition}
  Let $\phi$ be a pH-sentence with $m$ universal variables. Let $\mathscr{A}$ be an adversary and $\Omega_m$ a set of adversaries, both of length $m$.
  
  If $\mathcal{A}\models \phi_{\restrict\Omega_m}$ and  $\mathscr{A} \reactivelycomposable \Omega_m$ then
  $\mathcal{A} \models \phi$.
\end{theorem}
\begin{proof}
   We sketch the proof for the sake of completeness.
   Let $\Omega_m:=\{\mathscr{B}_1,\ldots,\mathscr{B}_k\}$ and $f$ and
   $g^i_j$ be as in the definition of reactive composition and
   witnessing that $\mathscr{A} \reactivelycomposable f(\mathscr{B}_1,\ldots,\mathscr{B}_k)$. 
   Assume also that $\mathcal{A}\models \phi_{\restrict\Omega_m}$.
   Given any sequence of play of the universal player according to the adversary $\mathscr{A}$, that is $v_1$ is played as $a_1 \in A_1$, $v_2$ is played as $a_2 \in A_2$, etc., we  ``go backwards through $f$'' via the maps $g^i_j$ to pinpoint \emph{incrementally} for each $j \in [k]$ a sequence of play $v_1=g^1_j(a_1)$, $v_2=g^2_j(a_1,a_2)$ etc, thus yielding eventually a tuple that belongs to adversary $\mathscr{B}_j$. After each block of universal variables, we lookup the winning strategy for the existential player against each adversary $\mathscr{B}_j$ and ``going forward through $f$'', that is applying $f$ to the choice of values for an existential variable against each adversary, we obtain a consistent choice for this variable against adversary $\mathscr{A}$ (this is because $f$ is a polymorphism and the quantifier-free part of the sentence $\phi$ is conjunctive positive). Going back and forth we obtain eventually an assignment to the existential variables that is consistent with the universal variables being played as $a_1,a_2,\ldots,a_m$.
\end{proof}
As a concrete example of an interesting sequence of adversaries, consider the adversaries for the notion of
\emph{$p$-collapsibility}. Let $p\geq 0$ be some fixed integer.
For $x$ in $A$, let $\Upsilon_{m,p,x}$ be the set of all 
rectangular 
adversaries of length $m$ with $p$ coordinates that are the set $A$
and all the other that are the fixed singleton $\{x\}$. For
$B\subseteq A$, let $\Upsilon_{m,p,B}$ be the union of $\Upsilon_{m,p,x}$ for all $x$ in $B$.
Let $\mathbf\Upsilon_{p,B}$ be the sequence of adversaries
$\Bigl(\Upsilon_{m,p,B} \Bigr)_{m \in \integerset}$.
Chen's original definition \cite{hubie-sicomp} for a structure $\mathcal{A}$ to be \emph{$p$-collapsible from source $B$} was that for every $m$ and for all pH-sentence $\phi$
with $m$ universal variable, $\mathcal{A}\models \phi_{\restrict\Upsilon_{m,p,B}}$ implies   $\mathscr{A} \models \phi$.

Let us consider now the  adversaries for the notion of
\emph{$p$-switchability}. Let $p\geq 0$ be some fixed integer.
Let $\Xi_{m,p}$ be the set of all tuples $(x_1,\ldots,x_m)$ in which there exists $a_i<\ldots<a_{k'}$, for $k'\leq p$, so that
\[ (x_1,\ldots,x_m) = (x_1,\ldots,x_{a_1},x_{a_1+1},\ldots,x_{a_2},x_{a_2+1},\ldots,\ldots,x_{a_k'},x_{a_k'+1},\ldots,x_m),\]
where $x_1=\ldots=x_{a_1-1}$, $x_{a_1}=\ldots=x_{a_2-1}$, \ldots, $x_{a_{k'}}=\ldots=x_{a_m}$.  
Let $\mathbf\Xi_{p}$ be the sequence of adversaries
$\Bigl(\Xi_{m,p} \Bigr)_{m \in \integerset}$.
Chen originally defined \cite{AU-Chen-PGP} a constraint language $\mathcal{A}$ to be \emph{$p$-switchable} iff for every $m$ and for all pH-sentences $\phi$
with $m$ universal variable, $\mathcal{A}\models \phi_{\restrict\Xi_{m,p}}$ implies   $\mathscr{A} \models \phi$. We will contrast the different definitions once again in the key forthcoming theorem ``In Abstracto'' (Theorem \ref{MainResult:InAbstracto}), where we will finally prove them equivalent.

\subsection{The $\Pi_2$-case}
\label{sec:pi2}
For a $\Pi_2$-pH sentence, the existential player knows the values of
all universal variables beforehand, and it suffices for her to have a
winning strategy for each instantiation (and perhaps no way to
reconcile them as should be the case for an arbitrary sentence). 
This also means that considering a set of adversaries of same length
is not really relevant in this $\Pi_2$-case as we may as well consider
the union of these adversaries or the set of all their tuples .
\begin{lemma}[\textbf{principle of union}]
  \label{lemma:union}
  Let $\Omega_m$ be a set of adversaries of length $m$ and $\phi$ a $\Pi_2$-sentence with $m$ universal variables.
  Let $\mathscr{O}_{\cup\Omega_m}:=\bigcup_{\mathscr{O}\in \Omega_m} \mathscr{O}$
  and $\Omega_{\text{tuples}}:=\{\{t\} | t \in \mathscr{O}_{\cup\Omega_m}\}$. 
  We have the following equivalence.
  $$\mathcal{A}\models \varphi_{\restrict\Omega_m} 
  \quad \iff \quad
  \mathcal{A}\models \varphi_{\restrict\mathscr{O}_{\cup\Omega_m}}
  \quad \iff \quad
  \mathcal{A}\models \varphi_{\restrict\Omega_{\text{tuples}}}$$
\end{lemma}
\noindent The forward implications
$$\mathcal{A}\models \varphi_{\restrict\Omega_m} 
\quad \implies \quad
\mathcal{A}\models \varphi_{\restrict\mathscr{O}_{\cup\Omega}}
\quad \implies \quad
\mathcal{A}\models \varphi_{\restrict\Omega_{\text{tuples}}}$$
of Lemma~\ref{lemma:union} hold clearly for arbitrary pH-sentences. 
The proof is trivial and is a direct consequence of the following
obvious fact.
\begin{fact}
  \label{fac:union}
  Let $\Omega_m$ be a set of adversaries of length $m$ and $\phi$ a $\Pi_2$-sentence with $m$ universal variables.
  $$\mathcal{A}\models \varphi_{\restrict\Omega_m}$$
  $$\quad \Updownarrow \quad$$
  $$\forall \mathscr{O} \in \Omega_m \forall t=(a_1,\ldots,a_m)\in  \mathscr{O} \ \mathcal{A}\models \varphi_{\restrict\{t\}}$$
\end{fact}

\begin{remark}[following Lemma~\ref{lemma:union}]
  For a sentence that is not $\Pi_2$, this does not necessarily hold. For example, consider $\forall x \forall y \exists z \forall w \ E(x,z) \wedge E(y,z) \wedge E(w,z)$ on the irreflexive $4$-clique $\mathcal{K}_4$. The sentence is not true, but for all individual tuples $(x_0,y_0,w_0)$, we have $\exists z \ E(x_0,z) \wedge E(y_0,z) \wedge E(w_0,z)$.
\end{remark}

Let $\mathscr{A}$ be an adversary and $\Omega_m$ a set of adversaries, both of length $m$.
We say that $\Omega_m$ \emph{generates} $\mathscr{A}$ iff for any tuple $t$ in $\mathscr{A}$, there exists a $k$-ary polymorphism $f_t$ of
$\mathcal{A}$ and tuples $t_1,\ldots,t_k$ in $\Omega_{\text{tuples}}$
such that $f_t(t_1,\ldots,t_k)=t$.
We have the following analogue of Theorem~\ref{thm:hubieReactiveComposition}.
\begin{proposition}
  \label{prop:pi2composition}
  Let $\phi$ be a $\Pi_2$-pH-sentence with $m$ universal variables. Let $\mathscr{A}$ be an adversary and $\Omega_m$ a set of adversaries, both of length $m$.
  
  If $\mathcal{A}\models \phi_{\restrict\Omega_m}$ and  $\Omega_m$
  generates $\mathscr{A}$  then
  $\mathcal{A} \models \phi_{\restrict\mathscr{A}}$.
\end{proposition}

\begin{proof}
  The hypothesis that $\Omega_m$ generates $\mathscr{A}$
  can be rephrased as follows : for each tuple $t$ in $\mathscr{A}$,
  $\{t\}\reactivelycomposable f_t(t_1,t_2,\ldots,t_k)$, where
  $t_1,t_2,\ldots,t_k$ belong to $\Omega_{\text{tuples}}$. 
  To see this, it remains to note that the suitable $g^j_i$'s from the
  definition of composition are induced trivially as there is no
  choice: for every $j$ in $[k]$ and every $i$  
  in $[m]$ pick $g^j_i(a_1,a_2,\ldots,a_i)= t_{i,j}$ where $t_{i,j}$ is the $i$th
  element of $t_j$. 
  So by Theorem~\ref{thm:hubieReactiveComposition}, if
  $\mathcal{A}\models \phi_{\restrict\Omega_{\text{tuples}}}$ then
  $\mathcal{A}\models \phi_{\restrict\{t\}}$.
  As this holds for any tuple $t$ in $\mathscr{A}$, via the principle of union, it
  follows that $\mathcal{A} \models \phi_{\restrict\mathscr{A}}$.
\end{proof}

We will construct a \emph{canonical $\Pi_2$-sentence to assert that an adversary is
  generating.} Let $\mathscr{O}$ be some adversary of length $m$.
Let $\sigma^{(m)}$ be the signature $\sigma$ expanded with a sequence of $m$ constants. For a map $\mu$ from $[m]$ to $A$, we write $\mu\in \mathscr{O}$ as shorthand for $(\mu(1),\mu(2),\ldots,\mu(m))\in \mathscr{O}$.
For some set $\Omega_m$ of adversaries of length $m$, we consider the following $\sigma^{(m)}$-structure: 
$$\bigotimes_{\mathscr{O}\in \Omega_m} \bigotimes_{\mu \in \mathscr{O}} \mathfrak{A}_{\mu}$$
where the $\sigma^{(m)}$-structure $\mathfrak{A}_{\mu}$ denotes the expansion of $\mathcal{A}$ by $m$ constants as given by the map $\mu$.
Let $\phi_{\Omega_m,\mathcal{A}}$ be the $\Pi_2$-pH-sentence\footnotemark{} created
from the canonical query of the $\sigma$-reduct of this
$\sigma^{(m)}$-structure with the $m$ constants $c_{j}$ becoming
variables $w_{j}$, universally quantified outermost, when all
\emph{constants are pairwise distinct}.
\footnotetext{For two constraint languages $\mathcal{A}$ and $\mathcal{B}$, when
  $\Omega_m$ is $A^m$ and $m$ is $|A|^{B}$, $\mathcal{B}$ models this
  canonical sentence iff  $\QCSP(\mathcal{A}) \subseteq \QCSP(\mathcal{B})$~\cite{LICS2008}}
Otherwise, we will say that $\Omega_m$ is \emph{degenerate}, and not
define the canonical sentence.

Note that adversaries such as $\Upsilon_{m,p,B}$ corresponding to $p$-collapsibility 
are not degenerate for $p>0$, and degenerate for $p=0$.
\begin{proposition}
  \label{thm:characteringPi2}
  Let $\Omega_m$ be a set of adversaries of length $m$ that is not degenerate.
  The following are equivalent.
  \begin{romannum}
  \item for any $\Pi_2$-pH sentence $\psi$, $\mathcal{A}\models
    \psi_{\restrict\Omega_m}$ implies $\mathcal{A}\models \psi$.
    \label{pi2abstracto:logical:pi2}
  \item for any $\Pi_2$-pH sentence $\psi$, $\mathcal{A}\models \psi_{\restrict\mathscr{O}_{\cup\Omega}}$ implies $\mathcal{A}\models \psi$.
  \item for any $\Pi_2$-pH sentence $\psi$, $\mathcal{A}\models \psi_{\restrict\Omega_{\text{tuples}}}$ implies $\mathcal{A}\models \psi$.
  \item $\mathcal{A}\models
    \phi_{\mathscr{O}_{\cup\Omega},\mathcal{A}}$
    \label{pi2abstracto:canonical:pi2}
  \item $\mathcal{A}\models \phi_{\Omega_{\text{tuples}},\mathcal{A}}$
  \item $\Omega_m$ generates $A^m$.\label{pi2abstracto:algebraic:pi2}
  \end{romannum}
\end{proposition}
\begin{proof}
  The first three items are equivalent by Lemma~\ref{lemma:union} (these implications have the same conclusion and equivalent premises).
  The fourth and fifth items are trivially equivalent since $\phi_{\mathscr{O}_{\cup\Omega},\mathcal{A}}$ and $\phi_{\Omega_{\text{tuples}},\mathcal{A}}$ are the same sentence.

  We show the implication from the third item to the fifth. By
  construction, $\phi_{\Omega_{\text{tuples}},\mathcal{A}}$ is $\Pi_2$
  and it suffices to show that there exists a winning strategy for
  $\exists$  against any adversary $\{t\}$ in
  $\Omega_{\text{tuples}}$. This is true by construction.
  Indeed, note that there exists a winning strategy for $\exists$ in
  the $(\mathcal{A},\phi_{\Omega_{\text{tuples}},\mathcal{A}})$-game
  against adversary  $\{t\}$ iff   there is a homomorphism from the
  $\sigma^{(m)}$-structure $\bigotimes_{t'\in \Omega_{\text{tuples}}} \mathfrak{A}_{\mu_{t'}}$
  to  the $\sigma^{(m)}$-structure $\mathfrak{A}_{\mu_t}$, where
  $\mu_t:[m]\to A$ is the map induced 
  naturally by $t$. The projection is such a homomorphism.
  
  The penultimate item implies the last one: instantiate the universal
  variables of $\phi_{\Omega_{\text{tuples}},\mathcal{A}}$ as given by
  the $m$-tuple $t$ and pick for $f_t$ the homomorphism from the product
  structure witnessing that $\exists$ has a winning strategy. 

  Finally, the last item implies the first one by
  Proposition~\ref{prop:pi2composition}. 
\end{proof}

\subsection{The unbounded case}
\label{sec:unbounded}
Let $n$ denote the number of elements of the structure
$\mathcal{A}$. Let $\mathscr{B}$ be an adversary from
$\Omega_{n\cdot m}$. We will denote by $\text{Proj}\mathscr{B}$ the set of
adversaries of  length $m$ induced by projecting over some arbitrary
choice of $m$ coordinates, one in each block of size $n$; that is
$1\leq i_1\leq n,n+1\leq i_2\leq 2\cdot n, \ldots, n\cdot (m-1)+ 1\leq
i_m\leq n \cdot m$. 
Of special concern to us are \emph{projective sequences of adversaries}
$\mathbf{\Omega}$ satisfying the following  for every $m\geq 1$,
\begin{equation*}
  \forall \mathscr{B} \in \Omega_{n\cdot m} \,\,
  \exists \mathscr{A} \in \Omega_{m} \,\,
  \bigwedge_{\widetilde{\mathscr{B}} \in \text{Proj}\mathscr{B}} 
  \widetilde{\mathscr{B}} \subseteq \mathscr{A} \,\,\,\,
  \text{($m$-\textbf{projectivity})}
\end{equation*}
As an example, consider the adversaries for collapsibility.
\begin{fact} 
  \label{fact:collapsible:adversary:is:projective}
  Let $B \subseteq A$ and $p\geq 0$.
  The sequence of adversaries $\mathbf\Upsilon_{p,B}$ 
  are projective.
\end{fact}
\begin{example}
  For a concrete illustration consider $A=\{0,1,2\}$ (thus $n=3$).
  We illustrate the fact that $\mathbf\Upsilon_{p=2,B=\{0\}}$ is
  projective for $m=4$ and some adversary $\mathscr{B}\in \Omega_{n\cdot
      m}=\Upsilon_{p=2,B=\{0\},3\cdot 4=12}$. Adversaries are depicted
    vertically with horizontal lines separating the blocks.
  $${\scriptsize \begin{array}{c|cccc|c}
    \mathscr{B}\in \Omega_{n\cdot
      m}&\multicolumn{4}{c|}{\text{Proj}\mathscr{B}}& \mathscr{A} \in
    \Omega_{m}\\
    \hline
    \hline
    A &A          &A          &      &\xcancel{A} &  \\
    0 &\xcancel{0}&\xcancel{0}&\ldots&\xcancel{0} & A\\
    0 &\xcancel{0}&\xcancel{0}&      &0           &  \\
    \hline
    0 &0          &0          &      &\xcancel{0}&  \\
    0 &\xcancel{0}&\xcancel{0}&\ldots&\xcancel{0}& 0\\
    0 &\xcancel{0}&\xcancel{0}&      &0          &  \\
    \hline
    0 &0          &0          &      &\xcancel{0}&  \\
    0 &\xcancel{0}&\xcancel{0}&\ldots&\xcancel{0}& 0\\
    0 &\xcancel{0}&\xcancel{0}&      &0          &  \\
    \hline
    0 &0          &\xcancel{0}&      &\xcancel{0}&  \\
    A &\xcancel{A}&A          &\ldots&\xcancel{A}& A\\
    0 &\xcancel{0}&\xcancel{0}&      &0          &  \\
  \end{array}}$$
The adversary $\mathscr{A}$ dominates any adversary obtained by
projecting the original larger adversary $\mathscr{B}$ by keeping a
single position per block. 
\end{example}
We could actually consider \mbox{w.l.o.g.} sequences of \emph{singleton}
adversaries.
\begin{fact}
  If $\mathbf{\Omega}$ is projective then so is the sequence
  $\bigl(\bigcup_{\mathscr{O}\in\Omega_m} \mathscr{O}\bigr)_{m \in  \integerset}$.
\end{fact}

A \emph{canonical sentence for composability for arbitrary pH-sentences} with $m$ universal variables
may be constructed similarly to the canonical sentence for the $\Pi_2$ case, except that it will have $m.n$ universal variables, which we view as $m$ blocks of $n$ variables, where $n$ is the number of elements of the structure $\mathcal{A}$.
Let $\mathscr{O}$ be some adversary of length $m$.
Let $\sigma^{(n\cdot m)}$ be the signature $\sigma$ expanded with a sequence of $n\cdot m$ constants $c_{1,1},\ldots, c_{n,1}, c_{1,2} \ldots, c_{n,2}, \ldots c_{1,m} \ldots, c_{n,m}$. We say that a map $\mu$ from $[n]\times [m]$ to $A$ is \emph{consistent} with $\mathscr{O}$ iff for every $(i_1,i_2,\ldots,i_m)$ in $[n]^m$, the tuple $(\mu(i_1,1),\mu(i_2,2),\ldots,\mu(i_m,m))$ belongs to the adversary $\mathscr{O}$. We write $A^{[n\cdot m]}_{\restrict \mathscr{O}}$ for the set of such consistent maps.
For some set $\Omega_m$ of adversaries of length $m$, we consider the following $\sigma^{(n\cdot m)}$-structure: 
$$\bigotimes_{\mathscr{O}\in \Omega_m} \bigotimes_{\mu \in A^{[n\cdot m]}_{\restrict \mathscr{O}}} \mathfrak{A}_{\mathscr{O},\mu}$$
where the $\sigma^{(n\cdot m)}$-structure $\mathfrak{A}_{\mathscr{O},\mu}$ denotes the expansion of $\mathcal{A}$ by $n\cdot m$ constants as given by the map $\mu$.
Let $\phi_{n,\Omega_m,\mathcal{A}}$ be the $\Pi_2$-pH-sentence created from the canonical query of the $\sigma$-reduct of this $\sigma^{(n\cdot m)}$ product structure with the $n\cdot m$ constants $c_{ij}$ becoming variables $w_{ij}$, universally quantified outermost.
As for the canonical sentence of the $\Pi_2$-case, this sentence is
not well defined if constants are not pairwise distinct, which occurs precisely
for degenerate adversaries. 
\begin{lemma}
  \label{lem:CanonicalSentenceImpliesReactiveComposability}
  Let $\Omega_m$ be a set of adversaries of length $m$ that is not degenerate. Let $\mathcal{A}$ be a structure of size $n$.
  If $\mathcal{A}$ models $\phi_{n,\Omega_m,\mathcal{A}}$ then the full adversary $A^m$ is reactively composable from $\Omega_m$.
  That is,
  $\mathcal{A} \models \phi_{n,\Omega_m,\mathcal{A}} \quad \implies \quad A^m \reactivelycomposable \Omega_m$
\end{lemma}
\begin{proof}
  We let each block of $n$ universal variables of the canonical
  sentence $\phi_{n,\Omega_m,\mathcal{A}}$ enumerate the elements of $A$.
  That is, given an enumeration $a_1,a_2,\ldots,a_n$ of $A$, we set $w_{i,j}=a_i$ for every $j$ in $[m]$ and every $i$ in $[n]$.  

  The assignment to the existential variables provides us with a
  $k$-ary polymorphism (the sentence being built as the conjunctive
  query of a product of $k$ copies of $\mathcal{A}$) together with the
  desired partial maps. A coordinate $r$ in $[k]$ corresponds to a
  choice of some adversary $\mathscr{O}$ of $\Omega_m$ and some map
  $\mu_r$ from $[n]\times[m]$ to $A$, consistent with this adversary. The partial map $g^r_\ell:A^\ell\to A$ with $\ell$ in $[m]$ (and $r$ in $[k]$) is given by $\mu_r$ as follows:
  $g^r_\ell(a_{i_1},\ldots,a_{i_\ell})$ depends only on the last coordinate $a_{i_\ell}$ and takes value $\mu(i,\ell)$ if $a_{i_\ell}=a_i$.
  By construction of the sentence and the property of consistency of such $\mu_r$ with the adversary $\mathscr{O}$, these partial functions satisfy the properties as given in the definition of reactive composition.
\end{proof}
\begin{lemma} 
  Let $\mathbf{\Omega}$ be a sequence of sets of adversaries that has
  the $m$-projectivity property for some $m\geq 1$ such that
  $\Omega_{n\cdot m}$ is not degenerate. The following holds.
  \label{lem:AdversariesWinnableOnCanonicalSentence}
  \begin{romannum}
  \item 
    $\mathcal{A}\models \psi_{\restrict \Omega_{\mathbf{n\cdot m}}}, \text{ where } \psi={\phi_{n,\Omega_\mathbf{m},\mathcal{A}}}$
  \item If for every $\Pi_2$-sentence $\psi$ with $m.n$ universal variables, it holds that $\mathcal{A}\models \psi_{\restrict \Omega_{\mathbf{m.n}}}$ implies $\mathcal{A}\models \psi$, then
    $\mathcal{A}\models \phi_{n,\Omega_\mathbf{m},\mathcal{A}}$.
  \end{romannum}
\end{lemma}

\begin{proof}
  The second statement is a direct consequence of the first one.
  The proof of the first statement generalises an argument used in the proof of Proposition~\ref{thm:characteringPi2}.
  Consider any adversary $\mathscr{O}$ in $\Omega_{n\cdot m}$. For convenience, we name the positions of this adversary in a similar fashion to the universal variables of the sentence, namely by a pair $(i,j)$ in $[n]\times [m]$. By projectivity, there exists an adversary $\mathscr{O}'$ in $\Omega_{m}$ which dominates any adversary $\tilde{\mathscr{O}}$ in $\text{Proj}\mathscr{O}$ (obtained by projecting over an arbitrary choice of one position in each of the $m$ blocks of size $n$).
  In the product structure underlying the formula $\phi_{n,\Omega_m,\mathcal{A}}$, we consider the following structure:
  $$\bigotimes_{\mu \in A^{[n\cdot m]}_{\upharpoonright \mathscr{O}'}} \mathfrak{A}_{\mathscr{O}',\mu}$$ 

  An instantiation of the universal variables of $\phi_{n,\Omega_m,\mathcal{A}}$ according to some tuple $t$ from the adversary $\mathscr{O}$ corresponds naturally to a map $\mu_t$ from $[n]\times [m]$ to $A$. Observe that our choice of $\mathscr{O}'$ ensures that this map $\mu_t$ is consistent with $\mathscr{O}'$.
  An instantiation of the universal variables by $\mu_t$ induces a $\sigma^{(n\cdot m)}$-structure $\mathfrak{A}_{\mu_t}$ and a winning strategy for $\exists$ amounts to a homomorphism from the product $\sigma^{(n\cdot m)}$-structure underlying the sentence to this $\mathfrak{A}_{\mu_t}$. Since the component $\mathfrak{A}_{\mathscr{O}',\mu_t}$ of this product structure is isomorphic to $\mathfrak{A}_{\mu_t}$, we may take for a homomorphism the corresponding projection. This shows that     $\mathcal{A}\models \psi_{\restrict \Omega_{\mathbf{n\cdot m}}}$ where $\psi={\phi_{n,\Omega_\mathbf{m},\mathcal{A}}}$.
\end{proof}

\begin{theorem}
  \label{theo:Pi2ToGeneral}
  Let $\mathbf{\Omega}$ be a sequence of sets of adversaries that has
  the $m$-projectivity property for some $m\geq 1$ such that
  $\Omega_{n\cdot m}$ is not degenerate.
  The following chain of implications holds
  $$
  \ref{Pi2ToGeneral:Pi2GoodApproximation}
  \implies
  \ref{Pi2ToGeneral:BigPi2Canonicalsentence}
  \implies 
  \ref{Pi2ToGeneral:ReactiveComposition}
  \implies 
  \ref{Pi2ToGeneral:GoodApproximation}
  $$
  where,
  \begin{romannum}
  \item For every $\Pi_2$-pH-sentence $\psi$ with $m.n$ universal variables, $\mathcal{A}\models \psi_{\restrict \Omega_{m.n}}$ implies $\mathcal{A}\models \psi$.\label{Pi2ToGeneral:Pi2GoodApproximation}
  \item $\mathcal{A}\models \phi_{n,\Omega_m,\mathcal{A}}$.\label{Pi2ToGeneral:BigPi2Canonicalsentence}
  \item $A^m \reactivelycomposable \Omega_m$.\label{Pi2ToGeneral:ReactiveComposition}
  \item For every pH-sentence $\psi$ with $m$ universal variables, $\mathcal{A}\models \psi_{\restrict \Omega_{m}}$ implies $\mathcal{A}\models \psi$.\label{Pi2ToGeneral:GoodApproximation}
  \end{romannum}
\end{theorem}
\begin{proof}
  The first implication holds by the previous lemma (second item of Lemma~\ref{lem:AdversariesWinnableOnCanonicalSentence}, this is the step where we use projectivity).
  The second implication is Lemma~\ref{lem:CanonicalSentenceImpliesReactiveComposability}. The last implication is Theorem~\ref{thm:hubieReactiveComposition}.
\end{proof}
Thus, in the projective case, when an adversary is good enough in the
$\Pi_2$-case, it is good enough in general. This can be characterised
logically via canonical sentences or ``algebraically'' in terms of
reactive composition or the weaker and more usual composition
property (see \ref{abstracto:algebraic:pi2}  below). 
  \begin{theorem}[\textbf{In abstracto}]\label{MainResult:InAbstracto}
    Let $\mathbf{\Omega}=\bigl(\Omega_m\bigr)_{m \in \integerset}$ be a projective sequence of adversaries,
    none of which are degenerate. 
    The following are equivalent.
  \begin{romannum}
  \item For every $m\geq 1$, for every pH-sentence $\psi$ with $m$
    universal variables, $\mathcal{A}\models \psi_{\restrict
      \Omega_{m}}$ implies $\mathcal{A}\models \psi$. 
    \label{abstracto:logical:general}
  \item For every $m\geq 1$, for every $\Pi_2$-pH-sentence $\psi$ with $m$ universal
    variables, $\mathcal{A}\models \psi_{\restrict \Omega_{m}}$
    implies $\mathcal{A}\models \psi$.
    \label{abstracto:logical:pi2}
  \item For every $m\geq 1$, $\mathcal{A}\models
    \phi_{n,\Omega_m,\mathcal{A}}$.
    \label{abstracto:canonical:general}
  \item For every $m\geq 1$, $\mathcal{A}\models
    \phi_{\mathscr{O}_{\cup\Omega},\mathcal{A}}$.
    \label{abstracto:canonical:pi2}
  \item For every $m\geq 1$, $A^m \reactivelycomposable \Omega_m$.
    \label{abstracto:algebraic:general}
  \item For every $m\geq 1$, $\Omega_m$ generates $A^m$.
    \label{abstracto:algebraic:pi2}
  \end{romannum}
\end{theorem}

\begin{proof}  
  Propositions~\ref{thm:characteringPi2} establishes the equivalence
  between  \ref{abstracto:logical:pi2}, \ref{abstracto:canonical:pi2} and
  \ref{abstracto:algebraic:pi2} for fixed values of $m$ (numbered
  there as \ref{pi2abstracto:logical:pi2}, \ref{pi2abstracto:canonical:pi2} and
  \ref{pi2abstracto:algebraic:pi2}, respectively).
  
  To lift these relatively trivial equivalences to the general case,
  the principle of our current proof no longer preserves the parameter $m$. The chain of
  implications of Theorem~\ref{theo:Pi2ToGeneral} translates here, once
  the parameter is universally quantified, to the chain of implications
  $$
  \ref{abstracto:logical:pi2}
  \implies 
  \ref{abstracto:canonical:general} 
  \implies 
  \ref{abstracto:algebraic:general}
  \implies 
  \ref{abstracto:logical:general}
  $$
  
  The fact that \ref{abstracto:logical:general} implies
  \ref{abstracto:logical:pi2} is trivial\footnotemark{}, which
  concludes the proof. 
  
  \footnotetext{We note in passing and for purely pedagogical reason that 
  the implication~\ref{abstracto:algebraic:general} to
  \ref{abstracto:algebraic:pi2} is also trivial, while the natural implication
  \ref{abstracto:canonical:general} to \ref{abstracto:canonical:pi2}
  will appear as an evidence to the reader once the definition of the canonical
  sentences is digested.}

\end{proof}
\begin{remark}
  The above equivalences can be read along two dimensions:
  \begin{center}
    \begin{tabular}[h]{lc|c}
      & general        & $\Pi_2$       \\
      \hline
      logical interpolation &
      \ref{abstracto:logical:general}&~\ref{abstracto:logical:pi2}\\ 
      \hline
      canonical sentences &
      \ref{abstracto:canonical:general}&~\ref{abstracto:canonical:pi2}\\ 
      \hline
      algebraic interpolation & 
      \ref{abstracto:algebraic:general}&
      \ref{abstracto:algebraic:pi2}\\ 
    \end{tabular}
  \end{center}
\end{remark}
Chen's original definitions of collapsibility and switchability correspond with item $(i)$, while the definitions given in the introduction correspond with item $(vi)$. For example, it is the formulation $(i)$ that provides Chen's original proof that switchability yields a QCSP in NP (Theorem 7.11 in \cite{AU-Chen-PGP}). In that same paper, the property of switchability as defined in the introduction is only shown to yield that the $m$-alternation-QCSP (allow only inputs in $\Pi_m$ prenex form, where $m$ is fixed) is in NP (Proposition 3.3 in \cite{AU-Chen-PGP}).
Let $\CSP_c(\mathcal{A})$ and $\QCSP_c(\mathcal{A})$ be the versions of $\CSP(\mathcal{A})$ and $\QCSP(\mathcal{A})$, respectively, in which constants naming the elements of $\mathcal{A}$ may appear in instances.
The following ostensibly generalises Theorem 7.11 \cite{AU-Chen-PGP} to effective and ``projective'' PGP, though we now know from \cite{ZhukGap2015} via Theorem \ref{MainResult:InAbstracto} that switchability explains all finite-domain algebra PGP.
\begin{corollary}\label{cor:EffectiveProjectivePGPImpliesQCSPInNP}
  Let $\mathcal{A}$ be a constraint language.  Let $\mathbf{\Omega}$ be a
  sequence of non degenerate adversaries that is effective, projective and
  polynomially bounded such that $\Omega_m$ generates $A^m$ for every
  $m\geq 1$.
  
  Let $\mathcal{A}'$ be the constraint language $\mathcal{A}$, possibly expanded
  with constants naming elements, at least one for each element that occurs in $\mathbf{\Omega}$.
  The problem $\QCSP(\mathcal{A})$ reduces in polynomial time to
  $\CSP(\mathcal{A}')$.
  In particular, if $\mathcal{A}$ has all constants, the problem $\QCSP_c(\mathcal{A})$ reduces in polynomial time to
  $\CSP_c(\mathcal{A})$.
\end{corollary}

 \begin{proof}
  To check whether a pH-sentence $\phi$ with $m$ universal variables
  holds in $\mathcal{A}$, by Theorem~\ref{MainResult:InAbstracto}, 
  we only need to check that $\mathcal{A}\models \phi_{\restrict
    \mathscr{B}}$ for every $\mathscr{B}$ in $\Omega_m$. 
  The reduction proceeds as in the proof
  of~\cite[Lemma~7.12]{AU-Chen-PGP}, which we outline
  here for completeness. 

  Pretend first that we reduce $\mathcal{A}\models \phi_{\restrict
    \mathscr{B}}$ to a collection of CSP instances, one for each tuple
  $t$ of $\mathscr{B}$, obtained by instantiation of the universal variables with the
  corresponding constants. If $x$ is an existential variable in
  $\phi$, let $x_t$ be the corresponding variable in the CSP instance
  corresponding to $t$. We will in fact enforce equality constraints
  via renaming of variables to
  ensure that we are constructing Skolem functions. For any two tuples
  $t$ and $t'$ in $\mathscr{B}$ that agree on their first $\ell$ coordinates, let
  $Y_\ell$ be the corresponding universal variables of $\phi$. For every
  existential variable $x$ such that $Y_x$ (the universally quantified
  variables of $\phi$ preceding $x$) is contained in $Y_\ell$, we identify
  $x_t$ with $x_{t'}$.
\end{proof}
Since Zhuk has proved that all cases of PGP in finite algebras come from switchability, the most important cases of In Abstracto (Theorem~\ref{MainResult:InAbstracto}) and  Corollary \ref{cor:EffectiveProjectivePGPImpliesQCSPInNP} involve the already introduced adversaries $\Xi_{m,p}$ (for some $p$) substituted for the placeholder $\Omega_m$. Note that when $p>0$, $\Xi_{m,p}$ is non-degenerate, and the sequence $\left( \Xi_{m,p} \right)_{m \in \mathbb{N}}$ is readily seen to be projective. We have resisted giving In Abstracto (Theorem~\ref{MainResult:InAbstracto}) only for switchability in order to emphasise that the proof comes alone from non-degenerate and projective. However, let us state its consequence nonetheless.
\begin{corollary}
 Let $\mathcal{A}$ be a finite constraint language, with constants naming all of its elements, so that $\Pol(\mathcal{A})$ is switchable. Then $\QCSP(\mathcal{A})$ reduces to a polynomial number of instances of  $\CSP(\mathcal{A})$ and is in \NP.
\label{cor:in-NP}
\end{corollary}
 
\subsection{Studies of collapsibility}
\label{sec:collapse}

Let $\mathcal{A}$ be a constraint language, $B\subseteq A$ and $p \geq 0$. Recall the structure
$\mathcal{A}$ is \emph{$p$-collapsible with source $B$} when for all
$m\geq 1$, for all pH-sentences $\phi$ with $m$ universal quantifiers, 
$\mathcal{A} \models \phi$ iff $\mathcal{A}\models \phi_{\restrict \Upsilon_{m,p,B}}$.
Collapsible structures are very important: to the best of our
knowledge, they are in fact the only examples of structures that enjoy a form of polynomial QCSP
to CSP reduction. This is different if one considers structures with infinitely many relations where the more
general notion of \emph{switchability} crops up~\cite{AU-Chen-PGP}.
Our abstract results of the previous section apply to both
switchability and collapsibility but we concentrate here on the latter. 
This  result applies since the underlying sequence of
adversaries are projective (see
Fact~\ref{fact:collapsible:adversary:is:projective}), as long as $p>0$
(non degenerate case).
\begin{corollary}[\textbf{In concreto}]
  \label{MainResult:InConcreto:Collapsibility}  
  Let $\mathcal{A}$ be a structure, $\emptyset \subsetneq B\subseteq A$ and $p>0$.
  The following are equivalent.
  \begin{romannum}
  \item $\mathcal{A}$ is $p$-collapsible from source $B$.
    \label{concreto:logicalcollapsibility:general}
  \item $\mathcal{A}$ is $\Pi_2$-$p$-collapsible from source $B$.
    \label{concreto:logicalcollapsibility:pi2}
  \item For every $m$, the structure $\mathcal{A}$ satisfies the
    canonical $\Pi_2$-sentence with $m\cdot |A|$ universal variables
    $\varphi_{n,\Upsilon_{m,p,B},\mathcal{A}}$.
  \item For every $m$, the structure $\mathcal{A}$ satisfies the
    canonical $\Pi_2$-sentence with $m$ universal variables
    $\varphi_{\mathscr{U},\mathcal{A}}$, where
    $\mathscr{U}=\bigcup_{\mathcal{O}\in \Upsilon_{m,p,B}} \mathcal{O}$.
  \item For every $m$, there exists a polymorphism $f$ of $\mathcal{A}$ witnessing
    that $A^m \reactivelycomposable \Upsilon_{m,p,B}$.
    \label{concreto:algebraiccollapsibility:general}
    \label{concreto:algebraiccollapsibility:iii}
  \item For every $m$, for every tuple $t$ in $A^m$, there is a
    polymorphism $f_t$ of $\mathcal{A}$ of arity $k$ at most $\binom{m}{p}\cdot |B|$ and
    tuples $t_1,t_2,\ldots,t_k$ in $\Upsilon_{m,p,B}$ such that
    $f_t(t_1,t_2,\ldots,t_k)=t$.
    \label{concreto:algebraiccollapsibility:pi2}
  \end{romannum}
\end{corollary}
\begin{remark}\label{rem:0collapsibility}
  When $p=0$, we obtain degenerate adversaries and this is due to the fact that if a QCSP is permitted equalities, then $0$-collapsibility can never manifest (think of $\forall x, y \ x=y$).
\end{remark}
Suppose $\mathcal{A}$ is expanded with constants naming all the elements. Then in \cite{hubie-sicomp}, Case~\ref{concreto:algebraiccollapsibility:general} of Corollary~\ref{MainResult:InConcreto:Collapsibility} is equivalent to $\Pol(\mathcal{A})$ being $p$-collapsible (in the algebraic sense). It is proved in \cite{hubie-sicomp} that if $\Pol(\mathcal{A})$, is $k$-collapsible (in the algebraic sense), then $\mathcal{A}$ is $k$-collapsible (in the relational sense). We note that Corollary~\ref{MainResult:InConcreto:Collapsibility} proves the converse, finally tying together the two forms of collapsibility (algebraic and relational) that appear in \cite{hubie-sicomp} .


We will now give an application of
Corollary~\ref{MainResult:InConcreto:Collapsibility}. For a sequence $\beta \in \{0,1\}^*$, of length $|\beta|$, let $\mathcal{P}_\beta$ be the undirected path on $|\beta|$ vertices such that the $i^{th}$ vertex has a loop iff the $i^{th}$ entry of $\beta$ is $1$ (we may say that the path $\mathcal{P}$ is \emph{of the form} $\beta$). A path $\mathcal{H}$ is \emph{quasi-loop-connected} if it is of either of the forms 
\begin{romannum}
\item $0^a 1^b \alpha$, for $b>0$ and some $\alpha$ with $|\alpha|=a$,
  or
  \label{def:qlc:i}
\item $0^a \alpha$, for some  $\alpha$ with $|\alpha|\in \{a,a-1\}$.
  \label{def:qlc:ii}
\end{romannum}
\noindent A path whose self-loops induce a connected component is further said to be \emph{loop-connected}.
\begin{application}
\label{app:FairmontHotelTrick}
  Let $\mathcal{A}$ be a partially reflexive path (no constants are present)
  that is quasi-loop connected then $\Pol(\mathcal{A})$ has the PGP. 
\end{application}
\begin{proof}
  Indeed, a partially reflexive path $\mathcal{A}$ that is quasi-loop
  connected has the same QCSP as a partially reflexive path that is
  loop-connected $\mathcal{B}$ \cite{DBLP:conf/cp/MadelaineM12} since
  for some $r_a>0$ there is a surjective homomorphism $g$ from
  $\mathcal{A}^{r_a}$ to $\mathcal{B}$ and for some $r_b>0$ there is a
  surjective homomorphism $h$ from $\mathcal{B}^{r_b}$ to
  $\mathcal{A}$ (see main result of~\cite{LICS2008}). We also know
  that $\mathcal{B}$ admits a majority polymorphism $m$
  \cite{QCSPforests} and is therefore $2$-collapsible from any
  singleton source \cite{hubie-sicomp}
  and that Theorem~\ref{MainResult:InConcreto:Collapsibility} holds
  for $\mathcal{B}$.  Pick some arbitrary element $a$ in $\mathcal{A}$
  such that there is some $b$ in $\mathcal{B}$ satisfying
  $g(a,a,\ldots,a)=b$. Use $b$ as a source for $\mathcal{B}$.

  We proceed to lift~\ref{concreto:algebraiccollapsibility:pi2} of
  Corollary~\ref{MainResult:InConcreto:Collapsibility} from
  structure $\mathcal{B}$ to $\mathcal{A}$, which we recall here for
  $\mathcal{B}$ : for every $m$, for every tuple $t$ in $B^m$, there
  is a polymorphism $f_t$ of $\mathcal{B}$ of arity $k$ and tuples
  $t_1,t_2,\ldots,t_k$ in $\Upsilon_{m,2,b}$ such that
  $f_t(t_1,t_2,\ldots,t_k)=t$.

  Let $g^{k}$ denote the surjective homomorphism from
  $(\mathcal{A}^{r_a})^k$ to $\mathcal{B}^k$ that applies $g$
  blockwise.  Going back from $t_i$ through $g$, we can find $r_a$
  tuples $t_{i,1},t_{i,2},\ldots,t_{i,r_a}$ all in $\Upsilon_{m,2,a}$
  (adversaries based on the domain of $\mathcal{A}$) such that
  $g(t_{i,1},t_{i,2},\ldots,t_{i,r_a})=t_i$.  Thus, we can generate
  any $\widetilde{t}$ in $\mathcal{B}$ via $f_{\widetilde{t}}\circ(g^k)$ from
  tuples of $\Upsilon_{m,2,a}$.

  Let $\hat{t}$ be now some tuple of $\mathcal{A}$. By surjectivity of
  $h$, let 
  $\widetilde{t_1},\widetilde{t_2},\ldots,\widetilde{t_{r_b}}$
  be tuples of $\mathcal{B}$ such that
  $h(\widetilde{t_1},\widetilde{t_2},\ldots,\widetilde{t_{r_b}})=\hat{t}$. 
  The polymorphism of $\mathcal{A}$ 
  $(f_{\widetilde{t_1}}\circ(g^k),f_{\widetilde{t_2}}\circ(g^k),\ldots,f_{\widetilde{t_{r_b}}}\circ(g^k))$
  shows that $\Upsilon_{m,2,a}$ generates $\hat{t}$.
  This shows that $\mathcal{A}$ is also 2-collapsible from a
  singleton source. 
\end{proof}

The last two conditions of Corollary~\ref{MainResult:InConcreto:Collapsibility} provide us with
a semi-decidability result: for each $m$, we may look for a particular
polymorphism \ref{concreto:algebraiccollapsibility:general} or several
polymorphisms \ref{concreto:algebraiccollapsibility:pi2}. Instead of a sequence of
polymorphisms, we now strive for a better algebraic
characterisation. We will only be able to do so for the special case
of a singleton source, but this is the only case hitherto found in nature.  

  Let $\mathcal{A}$ be a structure with a constant $x$ naming some element.
  Call a $k$-ary polymorphism of $\mathcal{A}$ such that $f$ is surjective when restricted at any position to $\{x\}$ a \emph{Hubie-pol in $\{x\}$}.
Chen uses the following lemma to show $4$-collapsibility of bipartite
graphs and disconnected graphs~\cite[Examples 1 and
2]{Meditations}. Though, we know via a direct
argument~\cite{CiE2006} that these examples are in
fact $1$-collapsible from a singleton source.
\begin{lemma}[Chen's lemma {\cite[Lemma
    5.13]{hubie-sicomp}}]
  \label{lem:ChensLemma}
  Let $\mathcal{A}$ be a structure with a constant $x$ naming some element, so that $\mathcal{A}$ has a Hubie-pol in $\{x\}$. Then $\mathcal{A}$ is $(k-1)$-collapsible from source $\{x\}$.
\end{lemma}
\begin{proof}

  We sketch the proof for pedagogical reasons.
  Via Corollary~\ref{MainResult:InConcreto:Collapsibility}, it suffices to show
  that for any $m$, $A^m$ is generated by $\Upsilon_{m,k-1,x}$
  (instead of the notion of reactive composition). 

  Consider adversaries of length $m=k$ for now, that is from $\Upsilon_{k,k-1,x}$. 
  If we apply $f$ to these $k$ adversaries, we generate the full adversary $A^k$. With a picture (adversaries are drawn as columns):
  $$f
  \begin{pmatrix}
    \{x\}  &A      &A      &\ldots & A      \\
    A      &\{x\}  &A      &\ldots & A      \\
    \vdots &       &\ddots &       & \vdots \\
    A      &\ldots & A     & \{x\} & A      \\
    A      &\ldots & A     & A     & \{x\} \\
  \end{pmatrix}
  = 
  \begin{pmatrix}
   A \\
   A \\
   \vdots \\
   A \\
   A \\
  \end{pmatrix}
  =
  A^k
  $$
  Expanding these adversaries uniformly with singletons $\{x\}$ to the
  full length $m$, we may produce an adversary from $\Upsilon_{m,k,x}$. 
  With a picture for \textsl{e.g.} trailing singletons:
  $$f
  \begin{pmatrix}
    \{x\}  &A      &A      &\ldots & A      \\
    A      &\{x\}  &A      &\ldots & A      \\
    \vdots &       &\ddots &       & \vdots \\
    A      &\ldots & A     & \{x\} & A      \\
    A      &\ldots & A     & A     & \{x\} \\
    \{x\}  &\{x\}  &\{x\}  &\ldots & \{x\}  \\
    \vdots &\vdots &\vdots &\vdots &\vdots  \\
    \{x\}  &\{x\}  &\{x\}  &\ldots & \{x\}  \\
  \end{pmatrix}
  = 
  \begin{pmatrix}
    A \\
    A \\
    \vdots \\
    A \\
    A \\
    \{x\} \\
    \vdots \\
    \{x\} \\
  \end{pmatrix}
  $$
  Shifting the first additional row of singletons in the top block, we
  will obtain the family of adversaries from $\Upsilon_{m,k,x}$ with
  a single singleton in the first $k+1$ positions. It should be now
  clear that we may iterate this process to derive $A^m$ eventually
  via some term $f'$ which is a superposition of $f$ and projections
  and is therefore also a polymorphism of $\mathcal{A}$.
\end{proof}

\begin{remark}
  An extended analysis of our proof should convince the careful reader that
  we may in the same fashion prove reactive composition (the polymorphism's
  action is determined for a row independently of the others). Thus, appealing to
  the previous section is not essential, though it does allow for a
  simpler argument. 
\end{remark}

An interesting consequence of last section's formal work is a form of
converse of Chen's Lemma, which allows us to give an algebraic
characterisation of collapsibility from a singleton source.
\begin{proposition}
  \label{prop:singletoncollapse:WeakCharacterisation}
  Let $x$ be a constant in $\mathcal{A}$.
  The following are equivalent:
  \begin{romannum}
  \item $\mathcal{A}$ is collapsible from $\{x\}$. 
  \item $\mathcal{A}$ has a Hubie-pol in $\{x\}$.
    \label{item:HubiePolx}
  \end{romannum}
\end{proposition}
\begin{proof}
  Lemma~\ref{lem:ChensLemma} shows that \ref{item:HubiePolx} implies
  collapsibility. We prove the converse.
  
  Assume $p$-collapsibility. By
  Fact~\ref{fact:collapsible:adversary:is:projective}, we may apply
  Theorem~\ref{MainResult:InAbstracto}. 
  For $m = p+1$, item \ref{abstracto:algebraic:general} of this theorem states that
  there is a polymorphism $f$ witnessing that $A^{p+1} \reactivelycomposable
  \Upsilon_{p+1,p,x}$ (diagrammatically, we may draw a similar picture
  to the one we drew at the beginning of the previous proof). Clearly,
  $f$ satisfies~\ref{item:HubiePolx}.
\end{proof}
\noindent In the proof of the above, for $(i)\Rightarrow (ii) \Rightarrow (i)$, we no longer control the collapsibility parameter as the
arity of our polymorphism is larger than the parameter we start with.
By inspecting more carefully the properties of the polymorphism $f$ we
get as a witness that $\mathcal{A}$ models a canonical sentence, we
may derive in fact $p$-collapsibility by an argument akin to the one
used above in the proof of Chen's Lemma. We obtain this way a nice
concrete result to counterbalance the abstract
Theorem~\ref{MainResult:InAbstracto}.
\begin{theorem}[\textbf{$p$-collapsibility from a singleton source}]
  \label{theo:singletonSource:StrongCharacterisation}
  Let $x$ be a constant in $\mathcal{A}$ and $p>0$.
  The following are equivalent.
  \begin{romannum}
  \item $\mathcal{A}$ is $p$-collapsible from $\{x\}$.
  \item For every $m\geq 1$, the full adversary $A^m$ is reactively composable from $\Upsilon_{m,p,x}$.
  \item $\mathcal{A}$ is $\Pi_2$-$p$-collapsible from $\{x\}$. 
  \item For every $m\geq 1$, $\Upsilon_{m,p,x}$ generates $A^m$.
  \item $\mathcal{A}$ models $\phi_{n,\Upsilon_{p+1,p,x},\mathcal{A}}$
    (which implies that $\mathcal{A}$ admits a particularly well
    behaved Hubie-pol in $\{x\}$ of arity $(p+1)n^p$).
  \end{romannum}
\end{theorem}

\begin{proof}
  Equivalence of the first four points appears in
   Corollary~\ref{MainResult:InConcreto:Collapsibility}, 
  as does the
  equivalence with the statement :
  For every $m \geq 1$, $\mathcal{A}$ models $\phi_{n,\Upsilon_{m,p,x},\mathcal{A}}$.
  So they imply trivially the last point by selecting $m = p+1$.
  
  We show that the last point implies the penultimate one. The proof principle is similar to that of Chen's Lemma.
  As we have argued similarly before, the last point implies the existence of a polymorphism $f$.
  This polymorphism enjoys the following property (each column represents in fact $n^p$ coordinates of $A$):
  $$f
  \begin{pmatrix}
    \begin{array}[h]{c|c|c|c|c}
    \{x\}   &A      &A      &\ldots & A      \\
    A      &\{x\}  &A      &\ldots & A      \\
    \vdots &       &\ddots &       & \vdots \\
    A      &\ldots & A     & \{x\} & A      \\
    A      &\ldots & A     & A     & \{x\} \\
  \end{array}
  \end{pmatrix}
  = 
  \begin{pmatrix}
   A \\
   A \\
   \vdots \\
   A \\
   A \\
  \end{pmatrix}
  =
  A^{p+1}
  $$
  So arguing as in the proof of Chen's Lemma, we may conclude similarly that 
  for all $m$, the full adversary $A^m$ is composable from $\Upsilon_{m,p,x}$.
\end{proof}

\begin{remark}\label{rem:collapsibilityOnConservativeSourceSet}
  We say that a structure $\mathcal{A}$ is \emph{$B$-conservative} where $B$ is a subset
  of its domain iff for any polymorphism $f$ of $\mathcal{A}$ and
  any $C \subseteq B$, we have $f(C,C,\ldots,C)\subseteq C$. 
  Provided that the structure is conservative on the source set $B$,
  we may prove a similar result for $p$-collapsibility from a conservative source.
\end{remark}

\subsubsection{The curious case of $0$-collapsibility}

Expanding on Remark~\ref{rem:0collapsibility}, we note that if we
forbid equalities in the input to a QCSP, then we can observe the
natural case of $0$-collapsibility, to which now we turn. This is not
a significant restriction in a context of complexity, since in all but
trivial cases of a one element domain, one can propagate equality out
through renaming of variables. 

We investigated a similar notion in the context of positive equality
free first-order logic, the syntactic restriction of first-order logic
that consists of sentences using only $\exists, \forall, \land$ and $\lor$.
For this logic, relativisation of quantifiers fully explains the complexity classification of the model
checking problem (a tetrachotomy between Pspace-complete, NP-complete,
co-NP-complete and Logspace)~\cite{LICS2011}. In
particular, a complexity in NP is characterised algebraically by the
preservation of the structure by a \emph{simple $A$-shop} (to be
defined shortly), which is equivalent to a strong form of
$0$-collapsibility since it applies not only to pH-sentences but also
to sentences of positive equality free first-order logic. 
We will show that this notion corresponds in fact to $0$-collapsibility
from a singleton source. Let us recall first some definitions.

A \emph{shop} on a set $B$, short for surjective hyper-operation, is a
function $f$ from $B$ to its powerset such that $f(x)\neq \emptyset$
for any $x$ in $B$ and for every $y$ in $B$, there exists $x$ in $B$
such that $f(x)\ni y$. An \emph{A-shop}\footnotemark{}
\footnotetext{The A does not stand for the name of the set, it is
  short for \emph{All}.} satisfies further that there is some $x$ such
that $f(x)=B$. A \emph{simple $A$-shop} satisfies further that
$|f(x')|=1$ for every $x'\neq x$.
We say that a shop $f$ is a \emph{she} of the structure $\mathcal{B}$, short for \emph{surjective
  hyper-endomorphism}, iff for any relational symbol $R$ in $\sigma$
of arity $r$, for any elements $a_1,a_2\ldots,a_r$ in $B$, if
$R(a_1,\ldots,a_r)$ holds in $\mathcal{B}$ then $R(b_1,\ldots,b_r)$
holds in $\mathcal{B}$ for any $b_1 \in f(a_1),\ldots,b_r \in f(a_r)$. 
We say that $\mathcal{B}$ \emph{admits a (simple) $A$-she} if there is a (simple) $A$-shop
$f$ that is a she of $\mathcal{B}$.

\begin{theorem}\label{theorem:0collapsibility}
  Let $\mathcal{B}$ be a finite structure.
  The following are equivalent.
  \begin{romannum}
  \item $\mathcal{B}$ is $0$-collapsible from source $\{x\}$ for some
    $x$ in $B$ for equality-free pH-sentences.\label{0collapsibility:singleton:pH}
  \item $\mathcal{B}$ admits a simple $A$-she.\label{0collapsibility:she}
  \item $\mathcal{B}$ is $0$-collapsible from source $\{x\}$ for some $x$ in $B$ for sentences of positive
    equality free first-order logic. \label{0collapsibility:mylogic}
  \end{romannum}
\end{theorem}
\begin{proof}
  The last two points are equivalent~\cite[Theorem
  8]{DBLP:conf/lics/MadelaineM09} (this result is stated with
  $A$-she rather than simple $A$-she but clearly, $\mathcal{A}$ has an A-she iff it has a simple A-she).
  The implication \ref{0collapsibility:she} to
  \ref{0collapsibility:singleton:pH} 
  follows trivially. 
  
  We prove the implication \ref{0collapsibility:singleton:pH} to
  \ref{0collapsibility:she} by contraposition.
  Assume that $A=[n]=\{1,\ldots,n\}$ and suppose that $\mathcal{A}$ has no simple A-she. We will prove that
  $\mathcal{A}$ does not admit universal relativisation to $x$
  for pH-sentences. We assume also \mbox{w.l.o.g.} that $x=1$.
  Let $\Xi$ be the set of simple A-shops $\xi$ \mbox{s.t.}
  $\xi(1)=[n]$. Since each $\xi$ is not a she of $\mathcal{A}$, we
  have a quantifier-free formula with $2n-1$ variables $R_\xi$ that
  consists of a single positive atom (not all variables need appear
  explicitly in this atom) such that $\mathcal{A} \models R_\xi(1,\ldots,1,2,\ldots,n)$\footnote{There are $n$ ones.}, but $\mathcal{A} \notmodels R_\xi(\xi^1,\ldots,\xi^n,\xi(2),\ldots,\xi(n))$ for some $\xi^1,\ldots,\xi^n \in [n]=\xi(1)$.

  \newcommand{\Eta}{\mathrm{E}}
  This means that for each $\eta:\{2,\ldots,n\} \rightarrow
  [n]$ there is some $2n-1$-ary ``atom'' $R_\eta$ such that
  $\mathcal{A} \models R_\eta(1,\ldots,1,1,2,\ldots,n)$\footnote{There
    are $n$ ones.}, but $\mathcal{A} \notmodels
  R_\eta(\xi^1,\ldots,\xi^n,\eta(2),\ldots,\eta(n))$ for some
  $\xi^1,\ldots,\xi^n \in [n]$. Let $\Eta=[n]^{[n-1]}$ denotes the set
  of $\eta$s.

  Suppose we had universal relativisation to $1$. Then we know that 
\[ \mathcal{A} \models \bigwedge_{\eta \in \Eta} R_\eta(1,\ldots,1,1,2,\ldots,n), \] that is, 
\[ \mathcal{A} \models \exists y_1,\ldots,y_n \bigwedge_{\eta \in \Eta} R_\eta(1,\ldots,1,y_1,y_2,\ldots,y_n).\]
According to relativisation this means also that 
\[ \mathcal{A} \models \exists y_1,\ldots,y_n \forall x_1, \ldots,x_n \bigwedge_{\eta \in \Eta} R_\eta(x_1,\ldots,x_n,y_1,y_2,\ldots,y_n).\]
But we know 
\[ \mathcal{A} \models \forall y_1,\ldots,y_n \exists x_1, \ldots,x_n \bigvee_{\eta \in \Eta} \neg R_\eta(x_1,\ldots,x_n,y_1,y_2,\ldots,y_n),\]
since the $\eta$s range over all maps $[n]$ to $[n]$. Contradiction.
\end{proof}
The above applies to singleton source only, but up to taking a power
of a structure (which satisfies the same QCSP), we may always place
ourselves in this singleton setting for $0$-collapsibility.
\begin{theorem}\label{theorem:0collapsibility:rainbowsource}
  Let $\mathcal{B}$ be a structure. The following are equivalent.
  \begin{romannum}
  \item $\mathcal{B}$ is $0$-collapsible from source $C$
  \item $\mathcal{B}^{|C|}$ is $0$-collapsible from some (any) singleton source $x$
    which is a (rainbow) $|C|$-tuple containing all elements of $C$.
  \end{romannum}
\end{theorem}

\begin{proof}
  Let $B=\{1,2,\ldots,b\}$.
  \begin{itemize}
  \item (downwards). 
    Let $x$ be $|B|$-tuple containing all elements of $B$, \mbox{w.l.o.g.} $x=(1,2,\ldots,b)$.
    Let $\varphi$ be a pH sentence.
    Assume that $\mathcal{A}^{|B|}\models \varphi_{\restrict
      (x,x,\ldots,x)}$.
    Equivalently, for any $i$ in $B$, $\mathcal{A}\models \varphi_{\restrict
      (i,i,\ldots,i)}$. Thus, $0$-collapsibility from source $B$
    implies that $\mathcal{A}\models \varphi$. Since $A$ and its power
    satisfy the same
    pH-sentences\cite{LICS2008}
    we may conclude that $\mathcal{A}^{|B|}\models \varphi$.
  \item (upwards). Assume that for any $i$ in $B$, $\mathcal{A}\models \varphi_{\restrict
      (i,i,\ldots,i)}$. Equivalently, $\mathcal{A}^{|B|}\models \varphi_{\restrict
      (x,x,\ldots,x)}$ where $x$ is any $|B|$-tuple containing all
    elements of $B$. By assumption, $\mathcal{A}^{|B|}\models \varphi$
    and we may conclude that $\mathcal{A}\models \varphi$.
  \end{itemize}
\end{proof}

\subsection{Issues of decidability}

The following is a corollary of Theorem \ref{theo:singletonSource:StrongCharacterisation}.
\begin{corollary}
  Given $p \geq 1$, a structure $\mathcal{A}$ and $x$ a constant in $\mathcal{A}$, we may decide
  whether $\mathcal{A}$ is $p$-collapsible from $\{x\}$.
\end{corollary}
\noindent We are not aware of a similar decidability result when the source is not a singleton. Neither are we aware of a decision procedure for collapsibility in general (when the $p$ is not specified).

The case of switchability in general can be answered by \cite{ZhukGap2015}. Let $\alpha,\beta$  be strict subsets of $A$ so that $\alpha \cup \beta = A$. A $k$-ary operation $f$ on $A$ is said to be \emph{$\alpha\beta$-projective} if there exists $i \in [k]$ so that $f(x_1,\ldots,x_k) \in \alpha$, if $x_i \in \alpha$, and $f(x_1,\ldots,x_k) \in \beta$, if $x_i \in \beta$. A constraint language $\mathcal{A}$, expanded with constants naming all the elements, is switchable iff there exists some $\alpha$ and $\beta$, strict subsets of $A$, so that $\alpha \cup \beta = A$ and some polymorphism of  $\mathcal{A}$ is not $\alpha\beta$-projective. If the maximal number of tuples in a relation of  $\mathcal{A}$ is $m$ then only polymorphisms of arity $m$ need be considered.

\section{The Chen Conjecture for infinite languages}

\subsection{NP-membership}

We need to revisit Theorem \ref{MainResult:InAbstracto} in the case of infinite  languages (signatures) and switchability. We omit parts of the theorem that are not relevant to us.
 \begin{theorem}[\textbf{In abstracto levavi}]\label{MainResult:InAbstractoLevavi}
    Let $\mathbf{\Omega}=\bigl(\Omega_m\bigr)_{m \in \integerset}$ be the sequence of the set of all ($k$-)switching $m$-ary adversaries over the domain of $\mathcal{A}$, a finite-domain structure with an infinite signature. The following are equivalent.
  \begin{romannum}
  \item[$(i)$] For every $m\geq 1$, for every pH-sentence $\psi$ with $m$ universal
    variables, $\mathcal{A}\models \psi_{\restrict \Omega_{m}}$
    implies $\mathcal{A}\models \psi$.
  \item[$(vi)$] For every $m\geq 1$, $\Omega_m$ generates $\mathrm{Pol}(\mathcal{A})^m$.
  \end{romannum}
\end{theorem}
\begin{proof}
We know from Theorem~\ref{MainResult:InAbstracto} that the following are equivalent.
 \begin{romannum}
  \item[$(i')$] For every finite-signature reduct $\mathcal{A}'$ of $\mathcal{A}$ and $m\geq 1$, for every pH-sentence $\psi$ with $m$ universal
    variables, $\mathcal{A}' \models \psi_{\restrict \Omega_{m}}$
    implies $\mathcal{A}' \models \psi$.
  \item[$(vi')$] For every finite-signature reduct $\mathcal{A}'$ of $\mathcal{A}$ and every $m\geq 1$, $\Omega_m$ generates $\mathrm{Pol}(\mathcal{A}')^m$.
  \end{romannum}
Since it is clear that both  $(i) \Rightarrow (i')$ and $(vi) \Rightarrow (vi')$, it remains to argue that $(i') \Rightarrow (i)$ and $(vi') \Rightarrow (vi)$.

[$(i') \Rightarrow (i)$.] By contraposition, if $(i)$ fails then it fails on some specific  pH-sentence $\psi$ which only mentions a finite number of relations of $\mathcal{A}'$. Thus $(i')$ also fails on some finite reduct of $\mathcal{A}'$ mentioning these relations.

[$(vi') \Rightarrow (vi)$.] Let $m$ be given. Consider some chain of finite reducts $\mathcal{A}_1,\ldots,\mathcal{A}_i,\ldots$ of $\mathcal{A}$ so that each $\mathcal{A}_i$ is a reduct of $\mathcal{A}_j$ for $i<j$ and every relation of $\mathcal{A}$ appears in some $\mathcal{A}_i$. We can assume from $(vi')$ that $\Omega_m$ generates $\mathrm{Pol}(\mathcal{A}_i)^m$, for each $i$. However, since the number of tuples $(a_1,\ldots,a_m)$ and operations mapping $\Omega_m$ pointwise to $(a_1,\ldots,a_m)$, witnessing generation in $\mathrm{Pol}(\mathcal{A}')^m$, is finite, the sequence of operations $(f^i_1,\ldots,f^i_{|A|^m})$ (where $f^i_j$ witnesses generation of the $j$th tuple in $A^m$) witnessing these must have an infinitely recurring element as $i$ tends to infinity. One such recurring element we call $(f_1,\ldots,f_{|A|^m})$ and this witnesses generation in $\mathrm{Pol}(\mathcal{A})^m$.
\end{proof}

Note that in $(vi') \Rightarrow (vi)$ above we did not need to argue uniformly across the different $(a_1,\ldots,a_m)$ and it is enough to find an infinitely recurring operation for each of these individually.

The following result is the infinite language counterpoint to Corollary \ref{cor:in-NP}, that follows from  Theorem \ref{MainResult:InAbstractoLevavi} just as Corollary \ref{cor:in-NP} followed from Theorem \ref{MainResult:InAbstracto}.
\begin{theorem}
\label{thm:easy}
Let $\mathbb{A}$ be an idempotent algebra on a finite domain $A$. If $\mathbb{A}$ satisfies PGP, then QCSP$(\mathrm{Inv}(\mathbb{A}))$ reduces to a polynomial number of instances of CSP$(\mathrm{Inv}(\mathbb{A}))$ and is in NP.
\end{theorem}

\subsection{co-NP-hardness}

Suppose there exist $\alpha,\beta$ strict subsets of $A$ so that $\alpha \cup \beta = A$, define the relation $\tau_k(x_1,y_1,z_1\ldots,x_k,y_k,z_k)$ by
\[ \tau_k(x_1,y_1,z_1\ldots,x_k,y_k,z_k):=\rho'(x_1,y_1,z_1) \vee \ldots \vee \rho'(x_k,y_k,z_k),\]
where $\rho'(x,y,z)=(\alpha \times \alpha \times \alpha) \cup (\beta \times \beta \times \beta)$. Strictly speaking, the $\alpha$ and $\beta$ are parameters of $\tau_k$ but we dispense with adding them to the notation since they will be fixed at any point in which we invoke the $\tau_k$. The purpose of the relations $\tau_k$ is to encode co-NP-hardness through the complement of the problem  (monotone) \emph{$3$-not-all-equal-satisfiability} (3NAESAT). Let us introduce also the important relations $\sigma_k(x_1,y_1,\ldots,x_k,y_k)$ defined by
\[ \sigma_k(x_1,y_1,\ldots,x_k,y_k):=\rho(x_1,y_1) \vee \ldots \vee \rho(x_k,y_k),\]
where $\rho(x,y)=(\alpha \times \alpha) \cup (\beta \times \beta)$.
\begin{lemma}
The relation $\tau_k$ is pp-definable in $\sigma_k$.
\label{lem:new-revision}
\end{lemma}
\begin{proof}
We will argue that $\tau_k$ is definable by the conjunction $\Phi$ of $3^k$ instances of $\sigma_k$ that each consider the ways in which two variables may be chosen from each of the $(x_i,y_i,z_i)$, i.e. $x_i\sim y_i$ or $y_i\sim z_i$ or $x_i\sim z_i$ (where $\sim$ is infix for $\rho$). We need to show that this conjunction $\Phi$ entails $\tau_k$ (the converse is trivial). We will assume for contradiction that $\Phi$ is satisfiable but $\tau_k$ not. In the first instance of $\sigma_k$ of $\Phi$ some atom must be true, and it will be of the form $x_i\sim y_i$ or $y_i\sim z_i$ or $x_i\sim z_i$. Once we have settled on one of these three, $p_i\sim q_i$, then we immediately satisfy $3^{k-1}$ of the conjunctions of $\Phi$, leaving $2\cdot 3^{k-1}$ unsatisfied. Now we can evaluate to true no more than one other among $\{x_i\sim y_i, y_i\sim z_i, x_i\sim z_i\} \setminus \{p_i\sim q_i\}$, without contradicting our assumptions. If we do evaluate this to true also, then we leave $3^{k-1}$ conjunctions unsatisfied.
Thus we are now down to looking at variables with subscript other than $i$ and in this fashion we have made the space one smaller, in total $k-1$. Now, we will need to evaluate in $\Phi$ some other atom of the form  $x_j\sim y_j$ or $y_j\sim z_j$ or $x_j\sim z_j$, for $j\neq i$. Once we have settled on at most two of these three then we immediately satisfy $3^{k-2}$ of the conjunctions remaining of $\Phi$, leaving $3^{k-2}$ still unsatisfied. Iterating this thinking, we arrive at a situation in which $1$ clause is unsatisfied after we have gone through all $k$ subscripts, which is a contradiction. 
\end{proof}
\begin{theorem}
\label{thm:hard}
Let $\mathbb{A}$ be an idempotent algebra on a finite domain $A$. If $\mathbb{A}$ satisfies EGP, then QCSP$(\mathrm{Inv}(\mathbb{A}))$ is co-NP-hard.
\end{theorem}
\begin{proof}
We know from Lemma 11 in \cite{ZhukGap2015} that there exist $\alpha,\beta$ strict subsets of $A$ so that $\alpha \cup \beta = A$ and the relation $\sigma_k$ is in $\mathrm{Inv}(\mathbb{A})$, for each $k \in \mathbb{N}$. From Lemma~\ref{lem:new-revision}, we know also that $\tau_k$  is in $\mathrm{Inv}(\mathbb{A})$, for each $k \in \mathbb{N}$.

We will next argue that $\tau_k$ enjoys a relatively small specification in DNF (at least, polynomial in $k$). We first give such a specification for $\rho'(x,y,z)$.
\[ \rho'(x,y,z):= \bigvee_{a,a',a'' \in \alpha} x=a \wedge y=a' \wedge z=a'' \vee \bigvee_{b,b',b'' \in \beta} x=b \wedge y=b' \wedge z=b''\]
which is constant in size when $A$ is fixed. Now it is clear from the definition that the size of $\tau_n$ is polynomial in $n$.

We will now give a very simple reduction from the complement of 3NAESAT to QCSP$(\mathrm{Inv}(\mathbb{A}))$. 3NAESAT is well-known to be NP-complete \cite{Papa} and our result will follow.

Take an instance $\phi$ of 3NAESAT which is the existential quantification of a conjunction of $k$ atoms $\mathrm{NAE}(x,y,z)$. Thus $\neg \phi$ is the universal quantification of a disjunction of $k$ atoms $x=y=z$. We build our instance $\psi$ of QCSP$(\mathrm{Inv}(\mathbb{A}))$ from $\neg \phi$ by transforming the quantifier-free part $x_1=y_1=z_1 \vee \ldots \vee x_k=y_k=z_k$ to $\tau_k=\rho'(x_1,y_1,z_1) \vee \ldots \vee \rho'(x_k,y_k,z_k)$.

($\neg \phi \in \mathrm{co\mbox{-}3NAESAT}$ implies $\psi \in \mathrm{QCSP}(\mathrm{Inv}(\mathbb{A}))$.) From an assignment to the universal variables $v_1,\ldots,v_m$ of $\psi$ to elements $x_1,\ldots,x_m$ of $A$, consider elements $x'_1,\ldots,x'_m \in \{0,1\}$ according to 
\begin{itemize}
\item $x_i \in \alpha \setminus \beta$ implies $x'_i=0$, 
\item $x_i \in \beta \setminus \alpha$ implies $x'_i=1$, and
\item $x_i \in \alpha \cap \beta$ implies we don't care, so \mbox{w.l.o.g.} say $x'_i=0$.
\end{itemize}
The disjunct that is satisfied in the quantifier-free part of $\neg \phi$ now gives the corresponding disjunct that will be satisfied in $\tau_k$.

($\psi \in \mathrm{QCSP}(\mathrm{Inv}(\mathbb{A}))$ implies $\neg \phi \in \mathrm{co\mbox{-}3NAESAT}$.) From an assignment to the universal variables $v_1,\ldots,v_m$ of $\neg \phi$ to elements $x_1,\ldots,x_m$ of $\{0,1\}$, consider elements $x'_1,\ldots,x'_m \in A$ according to 
\begin{itemize}
\item $x_i=0$ implies $x'_i$ is some arbitrarily chosen element in $\alpha \setminus \beta$, and
\item $x_i=1$ implies $x'_i$ is some arbitrarily chosen element in $\beta \setminus \alpha$.
\end{itemize}
The disjunct that is satisfied in $\tau_k$ now gives the corresponding disjunct that will be satisfied in the quantifier-free part of $\neg \phi$.
\end{proof}
\noindent The demonstration of co-NP-hardness in the previous theorem was inspired by a similar proof in \cite{BodirskyChenSICOMP}. Note that an alternative proof that $\tau_k$ is in $\mathrm{Inv}(\mathbb{A})$ is furnished by the observation that it is preserved by all $\alpha\beta$-projections (see \cite{ZhukGap2015}). We note surprisingly that co-NP-hardness in Theorem~\ref{thm:hard} is optimal, in the sense that some (but not all!) of the cases just proved co-NP-hard are also in co-NP.
\begin{proposition}
Let $\alpha,\beta$ be strict subsets of $A:=\{a_1,\ldots,a_n\}$ so that $\alpha \cup \beta = A$ and $\alpha \cap \beta \neq \emptyset$. Then QCSP$(A;\{\tau_k:k \in \mathbb{N}\},a_1,\ldots,a_n)$ is in co-NP.
\label{prop:coNP1}
\end{proposition}
\begin{proof}
Assume $|A|>1$, \mbox{i.e.} $n>1$ (note that the proof is trivial otherwise).
Let $\phi$ be an input to QCSP$(A;\{\tau_k:k \in \mathbb{N}\},a_1,\ldots,a_n)$. We will now seek to eliminate atoms $v=a$ ($a \in \{a_1,\ldots,a_n\}$) from $\phi$. Suppose $\phi$ has an atom $v=a$. If $v$ is universally quantified, then $\phi$ is false (since $|A|>1$). Otherwise, either the atom $v=a$ may be eliminated with the variable $v$ since $v$ does not appear in a non-equality relation; or $\phi$ is false because there is another atom $v=a'$ for $a\neq a'$; or $v=a$ may be removed by substitution of $a$ into all non-equality instances of relations involving $v$. This preprocessing procedure is polynomial and we will assume \mbox{w.l.o.g.} that $\phi$ contains no atoms $v=a$. We now argue that $\phi$ is a yes-instance iff $\phi'$ is a yes-instance, where $\phi'$ is built from $\phi$ by instantiating all existentially quantified variables as any $a \in \alpha \cap \beta$. The universal $\phi'$ can be evaluated in co-NP (one may prefer to imagine the complement as an existential $\neg \phi' $ to be evaluated in NP) and the result follows.
\end{proof}
In fact, this being an algebraic paper, we can even do better. Let $\mathcal{B}$ signify a set of relations on a finite domain but not necessarily itself finite. For convenience, we will assume the set of relations of $\mathcal{B}$ is closed under all co-ordinate projections and instantiations of constants. 
Call $\mathcal{B}$ \emph{existentially trivial} if there exists an element $c \in B$ (which we call a \emph{canon}) such that for each $k$-ary relation $R$ of $\mathcal{B}$ and each $i \in [k]$, and for every $x_1,\ldots,x_k \in B$, whenever $(x_1,\ldots,x_{i-1},x_i,x_{i+1},\ldots,x_k) \in R^{\mathcal{B}}$ then also $(x_1,\ldots,x_{i-1},c,x_{i+1},\ldots,x_k) \in R^{\mathcal{B}}$. We want to expand this class to \emph{almost existentially trivial} by permitting conjunctions of the form $v=a_i$ or $v=v'$ with relations that are existentially trivial. 
\begin{lemma}
Let $\alpha,\beta$ be strict subsets of $A:=\{a_1,\ldots,a_n\}$ so that $\alpha \cup \beta = A$ and $\alpha \cap \beta \neq \emptyset$. The set of relations pp-definable in $(A;\{\tau_k:k \in \mathbb{N}\},a_1,\ldots,a_n)$ is almost existentially trivial.
\end{lemma}
\begin{proof}
Consider a formula with a pp-definition in $(A;\{\tau_k:k \in \mathbb{N}\},a_1,\ldots,a_n)$. We assume that only free variables appear in equalities since otherwise we can remove these equalities by substitution. Now existential quantifiers can be removed and their variables instantiated as the canon $c$. Indeed, their atoms $\tau_n$ may now be removed since they will always be satisfied. Thus we are left with a conjunction of equalities and atoms $\tau_n$, and the result follows.
\end{proof}
\begin{proposition}
If $\mathcal{B}$ is comprised exclusively of relations that are almost existentially trivial, then QCSP$(\mathcal{B})$ is in co-NP under the \textbf{DNF encoding}.
\label{prop:coNP2}
\end{proposition}
\begin{proof}
The argument here is quite similar to that of Proposition~\ref{prop:coNP1} except that there is some additional preprocessing to find out variables that are forced in some relation to being a single constant or pairs of variables within a relation that are forced to be equal. In the first instance that some variable is forced to be constant in a $k$-ary relation, we should replace with the $(k-1)$-ary relation with the requisite forcing. In the second instance that a pair of variables are forced equal then we replace again the $k$-ary relation with a $(k-1)$-ary relation as well as an equality. Note that projecting a relation to a single or two co-ordinates can be done in polynomial time because the relations are encoded in DNF. After following these rules to their conclusion one obtains a conjunction of equalities together with relations that are existentially trivial. Now is the time to propagate variables to remove equalities (or find that there is no solution). Finally, when only existentially trivial relations are left, all remaining existential variables may be evaluated to the canon $c$.
\end{proof}
\begin{corollary}
Let $\alpha,\beta$ be strict subsets of $A:=\{a_1,\ldots,a_n\}$ so that $\alpha \cup \beta = A$ and $\alpha \cap \beta \neq \emptyset$. Then QCSP$(\mathrm{Inv}(\mathrm{Pol}(A;\{\tau_k:k \in \mathbb{N}\},a_,\ldots,a_n)))$ is in co-NP under the \textbf{DNF encoding}.
\label{cor:coNP2}
\end{corollary}
This last result, together with its supporting proposition, is the only time we seem to require the ``nice, simple'' DNF encoding, rather than arbitrary propositional logic. We do not require DNF for Proposition~\ref{prop:coNP1} as we have just a single relation in the signature for each arity and this is easy to keep track of. We note that the set of relations $\{\tau_k:k \in \mathbb{N}\}$ is not maximal with the property that with the constants it forms a co-clone of existentially trivial relations. One may add, for example, $(\alpha \times \beta) \cup (\beta \times \alpha)$.

The following, together with our previous results, gives the refutation of the Alternative Chen Conjecture.
\begin{proposition}
Let $\alpha,\beta$ be strict subsets of $A:=\{a_1,\ldots,a_n\}$ so that $\alpha \cup \beta = A$ and $\alpha \cap \beta \neq \emptyset$. Then, for each finite signature reduct $\mathcal{B}$ of $(A;\{\tau_k:k \in \mathbb{N}\},a_1,\ldots,a_n)$, QCSP$(\mathcal{B})$ is in NL.
\label{prop:finite-NL}
\end{proposition}
\begin{proof}
We will assume $\mathcal{B}$ contains all constants (since we prove this case gives a QCSP in NL, it naturally follows that the same holds without constants). Take $m$ so that, for each $\tau_i \in \mathcal{B}$, $i\leq m$. Recall from Lemma~\ref{lem:new-revision} that $\tau_i$ is pp-definable in $\sigma_i$. We will prove that the structure $\mathcal{B}'$ given by $(A;\{\sigma_k:k \leq m \},a_1,\ldots,a_n)$ admits a $(3m+1)$-ary near-unanimity operation $f$ as a polymorphism, whereupon it follows that $\mathcal{B}$ admits the same near-unanimity polymorphism. We choose $f$ so that all tuples whose map is not automatically defined by the near-unanimity criterion map to some arbitrary $a \in \alpha \cap \beta$. To see this, imagine that this $f$ were not a polymorphism. Then some $(3m+1)$ $m$-tuples in $\sigma_i$ would be mapped to some tuple not in $\sigma_i$ which must be a tuple $\overline{t}$ of elements from $(\alpha \setminus \beta) \cup (\beta \setminus \alpha)$. Note that column-wise this map may only come from $(3m+1)$-tuples that have $3m$ instances of the same element. By the pigeonhole principle, the tuple $\overline{t}$ must appear as one of the $(3m+1)$ $m$-tuples in $\sigma_i$ and this is clearly a contradiction.

It follows from \cite{hubie-sicomp} that QCSP$(\mathcal{B})$ reduces to a polynomially bounded ensemble of ${n \choose 3m} \cdot n \cdot n^{3m}$ instances CSP$(\mathcal{B})$, and the result follows.
\end{proof}

\subsection{The question of the tuple-listing encoding}

\begin{proposition}
Let $\alpha:=\{0,1\}$ and $\beta:=\{0,2\}$. Then, QCSP$(\{0,1,2\};\{\tau_k:k \in \mathbb{N}\},0,1,2)$ is in P under the \textbf{tuple-listing encoding}.
\label{prop:Chen-fails}
\end{proposition}
\begin{proof}
Consider an instance $\phi$ of this QCSP of size $n$ involving relation $\tau_m$ but no relation $\tau_k$ for $k>m$. The number of tuples in $\tau_m$ is $>3^m$. 
Following Proposition~\ref{prop:coNP1} together with its proof, we may assume that the instance is strictly universally quantified over a conjunction of atoms (involving also constants). Now, a universally quantified conjunction is true iff the conjunction of its universally quantified atoms is true. We can further say that there are at most $n$ atoms each of which involves at most $3m$ variables. Therefore there is an exhaustive algorithm that takes at most $O(n\cdot 3^{3m})$ steps which is $O(n^4)$.
\end{proof}
\noindent The proof of Proposition~\ref{prop:Chen-fails} suggests an alternative proof of Proposition~\ref{prop:finite-NL}, but placing the corresponding QCSP in P instead of NL.
Proposition~\ref{prop:Chen-fails} shows that Chen's Conjecture fails for the tuple encoding in the sense that it provides a language $\mathcal{B}$, expanded with constants naming all the elements, so that Pol$(\mathcal{B})$ has EGP, yet QCSP$(\mathcal{B})$ is in P under the tuple-listing encoding. However, it does not imply that the algebraic approach to QCSP violates Chen's Conjecture under the tuple encoding. This is because $(\{0,1,2\};\{\tau_k:k \in \mathbb{N}\},0,1,2)$ is not of the form Inv$(\mathbb{A})$ for some idempotent algebra $\mathbb{A}$. For this stronger result, we would need to prove QCSP$(\mathrm{Inv}(\mathrm{Pol}(\{0,1,2\};\{\tau_k:k \in \mathbb{N}\},0,1,2)))$ is in P under the tuple-listing encoding. However, such a violation to Chen's Conjecture under the tuple-listing encoding is now known from \cite{ZhukM20}.

\section{Switchability, collapsability and the three-element case} 

An algebra $\mathbb{A}$ is a \emph{G-set} if its domain is not one-element and every of its operations $f$ is of the form $f(x_1, \ldots , x_k) = \pi(x_i)$ where $i \in [k]$ and $\pi$ is a permutation on A. An algebra $\mathbb{A}$ contains a G-set as a \emph{factor} if some homomorphic image of a subalgebra of $\mathbb{A}$ is a G-set. A \emph{Gap Algebra} \cite{hubie-sicomp} is a three-element idempotent algebra that omits a G-set as a factor and is not collapsible.

Let $f$ be a $k$-ary idempotent operation on domain $D$. We say $f$ is a \emph{generalised Hubie-pol} on $z_1\ldots z_k$ if, for each $i \in k$, $f(D,\ldots,D,z_i,D,\ldots,D)=D$ ($z_i$ in the $i$th position). Recall that when $z_1=\ldots=z_k=a$ this is called a Hubie-pol in $\{a\}$ and gives $(k-1)$-collapsibility from source $\{a\}$. In general, a generalised Hubie-pol does not bestow collapsibility (\mbox{e.g.} Chen's $4$-ary switchable operation $r$, below).  The name Hubie operation was used in \cite{LICS2015} for Hubie-pol and the fact that this leads to collapsibility is noted in \cite{hubie-sicomp}.

For this section  $\mathbb{A}$ is an idempotent algebra on a $3$-element domain $\{0,1,2\}:=D$. Assume $\mathbb{A}$ has precisely two subalgebras on domains $\{0,2\}$ and $\{1,2\}$ and contains the idempotent semilattice-without-unit operation $s$ which maps all tuples off the diagonal to $2$. Thus, $\mathbb{A}$ is a \emph{Gap Algebra} as defined in \cite{AU-Chen-PGP}. Note that the presence of $s$ removes the possibility to have a $G$-set as a factor. We say that $\mathbb{A}$ is $\{0,2\}\{1,2\}$-projective if for each $k$-ary $f$ in $\mathbb{A}$ there exists $i \leq k$ so that, if $x_i \in \{0,2\}$ then $f(x_1,\ldots,x_k) \in \{0,2\}$ and if $x_i \in \{1,2\}$ then $f(x_1,\ldots,x_k) \in \{1,2\}$. Let us now further assume that $\mathbb{A}$ is not $\{0,2\}\{1,2\}$-projective. This rules out the Gap Algebras that have EGP and we now know that $\mathbb{A}$ is switchable \cite{AU-Chen-PGP}. We will now consider the $4$-ary operation $r$ defined by Chen in \cite{AU-Chen-PGP}. 
\[
\begin{array}{ccc}
0111 & & 1 \\
1011 & r & 1 \\
0001 &\mapsto & 0 \\
0010 & & 0 \\
\mbox{else} & & 2.
\end{array} 
\]
Chen proved that $(D;r,s)$ is $2$-switchable but not $k$-collapsible, for any $k$ \cite{AU-Chen-PGP}. Let $f$ be a $k$-ary operation in $\mathbb{A}$ that is not $\{0,2\}\{1,2\}$-projective. Violation of $\{0,2\}\{1,2\}$-projectivity in $f$ means that for each $i \in [k]$ either 
\begin{itemize}
\item there is $x_i \in \{0,1\}$ and $x_1,\ldots,x_{i-1},x_{i+1},\ldots,x_k \in \{0,1,2\}$ so that $f(x_1,\ldots,x_k)=y \in (\{0,1\}\setminus \{x_i\})$, or
\item or $x_i=c$ and there is $x_1,\ldots,x_{i-1},x_{i+1},\ldots,x_k \in \{0,1,2\}$ so that $f(x_1,\ldots,x_k)=y \in \{0,1\}$.
\end{itemize}
Note that we can rule out the latter possibility and further assume $x_1,\ldots,x_{i-1},$ $x_{i+1},\ldots,x_k \in \{0,1\}$, by replacing $f$ if necessary by the $2k$-ary $f(s(x_1,x'_1),\ldots,$ $s(x_k,x'_k))$. Thus, we may assume that for each $i \in [k]$ there is $x_i \in \{0,1\}$ and $x_1,\ldots,x_{i-1},x_{i+1},\ldots,x_k \in \{0,1\}$ so that $f(x_1,\ldots,x_k)=y \in (\{0,1\}\setminus \{x_i\})$.

We wish to partition the $k$ co-ordinates of $f$ into those for which violation of $\{0,2\}\{1,2\}$-projectivity, on words in $\{0,1\}^k$:
\begin{itemize}
\item[$(i)$] happens with $0$ to $1$ but never $1$ to $0$.
\item[$(ii)$] happens with $1$ to $0$ but never $0$ to $1$.
\item[$(iii)$] happens on both $0$ to $1$ and $1$ to $0$.
\end{itemize}
Note that Classes $(i)$ and $(ii)$ are both non-empty (Class $(iii)$ can be empty). This is because if Class $(i)$ were empty then  $f(s(x_1,x'_1),\ldots,s(x_k,x'_k))$ would be a Hubie-pol in $\{1\}$ and if Class $(ii)$ were empty we would similarly have a Hubie-pol in $\{0\}$. We will write $k$-tuples with vertical bars to indicate the split between these classes. Suppose there exists a $\overline{z}$ so that $f(0,\ldots,0|1,\ldots,1|\overline{z}) \in \{0,1\}$. Then we can identify all the variables in one among Class $(i)$ or Class $(ii)$ to obtain a new function for which one of these classes is of size one. Note that if, e.g., Class $(i)$ is made singleton, this process may move variables previously in Class $(iii)$ into Class $(ii)$, but never to Class $(i)$. 

Thus we may assume that either Class $(i)$ or Class $(ii)$ is singleton or, for all $\overline{z}$ over $\{0,1\}$, $f(0,\ldots,0|1,\ldots,1|\overline{z}) =2$. Indeed, these singleton cases are dual and thus \mbox{w.l.o.g.} we need only prove one of them. Recall the global assumptions are in force for the remainder of the paper.

\subsection{Properties of Gap Algebras that are switchable}

\begin{lemma}
\label{lem:fun}
Any algebra over $D$ containing $f$ and $s$ is either collapsible or has binary term operations $p_1$ and $p_2$ so that 
\begin{itemize}
\item $p_1(0,1)=1$ and $p_1(1,0)=p_1(2,0)=2$, \textbf{and}
\item $p_2(0,1)=0$ and $p_2(1,0)=p_2(1,2)=2$.
\end{itemize}
\end{lemma}
\begin{proof}
Consider a tuple $\overline{x}$ over $\{0,1\}$ that witnesses the breaking of $\{0,2\}\{1,2\}$-projectivity for some Class $(i)$ variable from $0$ to $1$; so $f(\overline{x})=1$. Let $\widetilde{x}$ be $\overline{x}$ with the $0$s substituted by $2$ and the $1$s substituted by $0$. If, for each such  $\overline{x}$ over $\{0,1\}$ that witnesses the breaking of $\{0,2\}\{1,2\}$-projectivity for each Class $(i)$ variable, we find $f(\widetilde{x})=0$, then $f(s(x_1,x'_1),\ldots,s(x_k,x'_k))$ is a Hubie-pol in $\{1\}$. Thus, for some such  $\overline{x}$ we find $f(\widetilde{x})=2$. By collapsing the variables according to the division of $\overline{x}$ and $\widetilde{x}$ we obtain a binary function $p_1$ so that $p_1(0,1)=1$ and $p_1(2,0)=2$. We may also see that $p_1(1,0)=2$, since Classes $(i)$ and $(ii)$ are non-empty.

Dually, we consider tuples $\overline{x}$ over $\{0,1\}$ that witnesses the breaking of $\{0,2\}\{1,2\}$-projectivity for Class $(ii)$ variables from $1$ to $0$ to derive a function $p_2$ so that  $p_2(0,1)=0$, $p_2(1,2)=p(1,0)=2$.
\end{proof}

\subsubsection{The asymmetric case: Class $(i)$ is a singleton and there exists $\overline{z} \in \{0,1\}^*$ so that $f(0|1,\ldots,1|\overline{z})=1$}

We will address the case in which Class $(i)$ is a singleton and there exists $\overline{z} \in \{0,1\}^*$ so that $f(0|1,\ldots,1|\overline{z})=1$ (the like case with Class $(ii)$ being singleton itself being dual).

\begin{proposition}
\label{prop:asymmetric}
Let $f$ be so that Class $(i)$ is a singleton and there exists $\overline{z} \in \{0,1\}^*$ so that $f(0|1,\ldots,1|\overline{z})=1$. Then, either $f$ generates a binary idempotent operation with $01 \mapsto 0$ and $02 \mapsto 2$, or any algebra on $D$ containing $f$ and $s$ is collapsible.
\end{proposition}
\begin{proof}
Let us consider the general form of $f$,
\[
\begin{array}{c|ccc|ccccc}
0& 1 & \cdots & 1 & z_0^0 & \cdots & z_0^{\ell'} & \ & 1 \\
0 & y^1_1 & \cdots & y^{k'}_1 & z_1^1 & \cdots & z_1^{\ell'} & \ & 0 \\
\vdots & \vdots & \cdots & \vdots & \vdots & \cdots & \vdots & \mapsto & \vdots \\
0 & y^1_{m'} & \cdots & y^{k'}_{m'} & z_{m'}^1 & \cdots & z_{m'}^{\ell'} & \ & 0 \\
\end{array}
\]
where the $y$s and $z$s are from $\{0,1\}$ and we can assume that each $(y^i_1,\ldots,y^i_{m'})$ contains at least one $1$ and also each $(z^i_1,\ldots,z^i_{m'})$ contains at least one $1$. For the latter assumption recall that in Class $(iii)$ we can always find some break of $\alpha\beta$-projectivity from $1$ to $0$. Note that by expanding what we previously called Class $(ii)$ we can build, by possibly identifying variables, a function $f'$ of the form
\[
\begin{array}{c|ccc|ccccc}
0& 1 & \cdots & 1 & 0 & \cdots & 0 & \ & 1 \\
0 & y^1_1 & \cdots & y^{k}_1 & z_1^1 & \cdots & z_1^{\ell} & \ & 0 \\
\vdots & \vdots & \cdots & \vdots  & \vdots & \cdots & \vdots & \mapsto & \vdots \\
0 & y^1_{m} & \cdots & y^{k}_{m} & z_{m}^1 & \cdots & z_{m}^{\ell} & \ & 0 \\
\end{array}
\]
where the $y$s and $z$s are from $\{0,1\}$ and we can assume each $(y^i_1,\ldots,y^i_{m})$ contains at least one $1$ and also each $(z^i_1,\ldots,z^i_{m})$ contains a least one $1$. Note that we do not claim the new vertical bars correspond to delineate between Classes $(i)$, $(ii)$ and $(iii)$ under their original definitions, since this is not important to us. We will henceforth assume that $f$ is in the form of $f'$.

Let $x^i_j$ (resp., $v^i_j$) be $0$ if $y^i_j$ (resp., $z^i_j$) is $0$, and be $2$ if $y^i_j$ (resp., $z^i_j$) is $1$. That is, $(x^1_{j}, \ldots,  x^{k}_{j}, v_{j}^1\ \ldots, v_{j}^{\ell})$ is built from $(y^1_{j}, \ldots,  y^{k}_{j}, z_{j}^1\ \ldots, z_{j}^{\ell})$ by substituting $1$s by $2$s. Suppose one of $f(0|x^1_1,\ldots,x^k_1| v_1^1\ \ldots, v_1^{\ell})$, \ldots, $f(0|x^1_m,\ldots,x^k_m | v_{m}^1\ \ldots, v_{m}^{\ell})$ is $2$. Then $f$ generates an idempotent binary operation with $01 \mapsto 0$ and $02 \mapsto 2$. Thus, we may assume that each of  $f(0|x^1_1,\ldots,x^k_1 | v_1^1\ \ldots, v_1^{\ell})$, \ldots, $f(0|x^1_m,\ldots,x^k_m|$ $v_m^1\ \ldots, v_m^{\ell}))$ is $0$. We now move to consider some cases.

(Case 1: $\ell=0$, \mbox{i.e.} there is nothing to the right of the second vertical bar.) From adversaries of the form $(\{0\}^M)$ and $(\{0,1\}^{m-1},\{1\}^{M-m+1})$ this supports construction of $(\{0,1\}^{m},\{1\}^{M-m})$ and all co-ordinate permutations. We illustrate this with the following diagram which makes some assumptions about the locations of the $1$s in each $(y^i_1,\ldots,y^i_{m})$; nonetheless it should be clear that the method works in general since there is at least one $1$ in $(y^i_1,\ldots,y^i_{m})$.
\[
\begin{array}{c|cccccc}
\{0\} & \{0,1\} & \{0,1\} & \cdots & \{0,1\} & & \{0,1\} \\
\vdots & \vdots & \vdots & \cdots & \vdots & & \{0,1\} \\
\{0\} & \{0,1\} & \{0,1\} & \cdots & \{0,1\} & & \{0,1\} \\
\{0\} & \{1\} & \{0,1\} & \cdots & \{0,1\} & & \{0,1\} \\
\{0\} & \{0,1\}  & \{1\} & \cdots & \{0,1\} & & \{0,1\} \\
\vdots & \vdots & \vdots & \cdots & \vdots & \mapsto & \{0,1\} \\
\{0\} & \{0,1\} & \{0,1\} & \cdots & \{1\} & & \{0,1\} \\
\{0\} & \{1\} & \{1\} & \cdots & \{1\} & & \{1\} \\
\vdots & \vdots & \vdots & \cdots & \vdots & & \{1\} \\
\{0\} & \{1\} & \{1\} & \cdots & \{1\} & & \{1\} \\
\end{array}
\]
Applying $s$ it is clear that the full adversary may be built from, for example, $(D^{m-1},\{0\}^{M-m+1})$ and $(D^{m-1},\{b\}^{M-m+1})$ which demonstrates $(m-1)$-collapsibility.

(Case 2: $\ell\geq 1$.) Here we consider what is $f(0|1,\ldots,1|1,\ldots,1)$. If this is $1$ then we can clearly reduce to the previous case. If it is $0$ then $f(s(x_1,x_1'),\ldots,$ $s(x_{k+\ell+1},x'_{k+\ell+1}))$ is a generalised Hubie-pol in both $00|11,\ldots,11|00,\ldots,00$ and $00|11,\ldots,11|11,\ldots,11$, and we are collapsible. This is because the composed function on these listed tuples gives $1$ and $0$, respectively, thus permitting to build adversaries of the form $(\{0,1\}^{k+\ell+2},\{0\}^{M-k-\ell-2})$ and  $(\{0,1\}^{k+\ell+2},\{1\}^{M-k-\ell-2})$ from adversaries of the form $(\{0,1\}^{k+\ell+1},\{0\}^{M-k-\ell-1})$ and  $(\{0,1\}^{k+\ell+1},$ $\{1\}^{M-k-\ell-1})$ (\mbox{cf.} Case 1).

Thus, we may assume $f(0|1,\ldots,1|1,\ldots,1)=2$. Using the fact that $f(s(x_1,x_1'),$ $\ldots,s(x_{k+\ell+1},x'_{k+\ell+1}))$ is a generalised Hubie-pol in $00|11\ldots 11|00 \ldots 00$ we can build (using $s$ and rather like in Case 1), from adversaries of the form $(\{0\}^{M})$ and  $(D^{(m-1)i},\{1\}^{M-(m-1)i})$, adversaries of the form $(D^{mi},\{1\}^{M-mi})$, and all co-ordinate permutations of this. Similarly, using the fact that $f(s(x_1,x_1'),\ldots,$ $s(x_{k+\ell+1},x'_{k+\ell+1}))$ is a generalised Hubie-pol in $00|11\ldots 11|11 \ldots 11$, we can build adversaries of the form $(D^{mi},\{2\}^{M-mi})$. 

(Case 2a: $f(0|2,\ldots,2|2,\ldots,2)=2$.) Consider again
\[
\begin{array}{c|ccc|ccccc}
0 & x^1_1 & \cdots & x^{k}_1 & v_1^1 & \cdots & v_1^{\ell} & \ & 0 \\
\vdots & \vdots & \cdots & \vdots  & \vdots & \cdots & \vdots & \mapsto & \vdots \\
0 & x^1_{m} & \cdots & x^{k}_{m} & v_{m}^1 & \cdots & v_{m}^{\ell} & \ & 0 \\
0 & 2 & \cdots & 2 & 2 & \cdots & 2 & & 2\\
\end{array}
\]
where each $(x^i_1,\ldots,x^i_{m})$ and $(v^i_1,\ldots,v^i_{m})$ contains at least one $2$. By amalgamating Classes $(ii)$ and $(iii)$ we obtain some function with the form
\[
\begin{array}{c|ccccc}
0 & u^1_1 & \cdots & x^{\nu}_1 & & 0 \\
\vdots & \cdots & \vdots & & & \vdots \\
0 & u^1_m & \cdots & x^{\nu}_m & & 0 \\
0 & 2 & \cdots & 2 & & 2\\
\end{array}
\]
where each $(u^i_1,\ldots,u^i_{m})$ is in $\{0,2\}^*$ and contains at least one $2$. From adversaries of the form $(D^{r+m-1}, \{2\}^{M-r-m+1})$ and $(\{0\}^{M})$ we can build $(D^{r},\{0,2\}^{M-r})$, and all co-ordinate permutations. We begin, pedagogically preferring to view some $D$s as $\{0,2\}$s,
\[
\begin{array}{c|cccccc}
\{0\} & D & D & \cdots & D & & D \\
\vdots & \vdots & \vdots & \cdots & \vdots & & D \\
\{0\} & D & D & \cdots & D & & D \\
\{0\} & \{2\}  &  \{0,2\} & \cdots &  \{0,2\} & & \{0,2\} \\
\{0\} &  \{0,2\}  & \{c\} & \cdots &  \{0,2\} & & \{0,2\} \\
\vdots & \vdots & \vdots & \cdots & \vdots & \mapsto & \{0,2\} \\
\{0\} &  \{0,2\}  &  \{0,2\} & \cdots & \{2\} & & \{0,2\} \\
\{0\} & \{2\}  & \{2\} & \cdots & \{2\} & & \{2\} \\
\vdots & \vdots & \vdots & \cdots & \vdots & & \{2\} \\
\{0\} & \{2\}  & \{2\} & \cdots & \{2\} & & \{2\} \\
\end{array}
\]
and follow with bottom parts of the form
\[
\begin{array}{c|ccccccccc}
\vdots & \vdots & \vdots & \cdots & \vdots & \mapsto & \{2\} \mbox{ or } \{0,2\} \\
\{0\} & \{0,2\}  & \{0,2\} & \cdots & \{0,2\} & & \{0,2\}. \\
\end{array}
\]
This now supports bootstrapping of the full adversary from adversaries of the form $(D^{m^2},\{0\}^{M-m^2})$, $(D^{m^2},$ $\{1\}^{M-m^2})$ and $(D^{m^2},\{2\}^{M-m^2})$. 

(Case 2b:  $f(0|2,\ldots,2|2,\ldots,2)=0$.) Here, from adversaries of the form $(D^{r+m-1}, \{2\}^{M-r-m+1})$ and $(\{0\}^{M})$ we can directly build $(D^{r+m-1}, \{0\}^{M-r-m+1})$.
\[
\begin{array}{c|cccccc}
\{0\} & D & D & \cdots & D  & & D \\
\vdots & \vdots & \vdots & \cdots & \vdots  & & D \\
\{0\} & D & D & \cdots & D  & & D \\
\{0\} & \{2\}  & \{2\} & \cdots & \{2\} & & \{0\} \\
\vdots & \vdots & \vdots & \cdots & \vdots  & \mapsto & \{0\} \\
\{0\} & \{2\}  & \{2\} & \cdots & \{2\}  & & \{0\} \\
\end{array}
\]
This now supports bootstrapping of the full adversary, similarly as in Case 2a (but slightly simpler).
\end{proof}

Let $\overline{x}:=x_1,\ldots,x_k$ and $\overline{y}:=y_1,\ldots,y_k$ be words over $\{0,1\} \ni x,y$. Let $\wedge(x,y)=0$ if $0 \in \{x,y\}$ and $1$ otherwise. Let $\vee(x,y)=1$ if $1 \in \{x,y\}$ and $0$ otherwise. This corresponds with considering $0$ as $\bot$ and $1$ as $\top$. Define $\wedge(\overline{x},\overline{y}):=(\wedge(x_1,y_1),\ldots,\wedge(x_k,y_k))$ and $\vee(\overline{x},\overline{y}):=(\vee(x_1,y_1),\ldots,\vee(x_k,y_k))$.  We are most interested in words 
\begin{itemize}
\item[A] $(\overline{x} | 0,\ldots,0 | \overline{z})$, such that $f(\overline{x} | 0,\ldots,0 | \overline{z})=0$, and for no $\overline{x}'\neq \overline{x}$ and $\overline{z}'$ over $\{0,1\}$ do we have  $(\overline{x}' | 0,\ldots,0 | \overline{z}')$ with $\vee(\overline{x},\overline{x}')=\overline{x}'$ so that $f(\overline{x}' | 0,\ldots,0 | \overline{z}')=0$.
\item[B] $(1,\ldots,1 | \overline{y} | \overline{z})$, such that $f(1,\ldots,1 | \overline{y} | \overline{z})=1$, and for no $\overline{y}'\neq \overline{y}$  and $\overline{z}'$ over $\{0,1\}$ do we have $(1,\ldots,1 | \overline{y}' | \overline{z}')$ with $\wedge(\overline{y},\overline{y}')=\overline{y}'$ so that $f(1,\ldots,1 | \overline{y}' | \overline{z}')=1$.
\end{itemize}
Such $\overline{x}$ and $\overline{y}$ are in a certain sense \emph{maximal}, but the sense of maximality is dual in Case B from Case A. $\overline{x}$ is maximal under inclusion for the number of $1$s it contains and $\overline{y}$ is maximal under inclusion for the number of $0$s it contains. In the asymmetric case that we consider here \mbox{w.l.o.g.}, only Case A above will be salient, but we introduce both now for pedagogical reasons.

\begin{lemma}
\label{lem:r4-asymmetric}
Let $f$ be so that Class $(i)$ is a singleton and there exists $\overline{z} \in \{0,1\}^*$ so that $f(0|1,\ldots,1|\overline{z})=1$. Then any algebra over $D$ containing $f$ and $s$ is either collapsible or has a $4$-ary term operation $r_4$ so that
\[
\begin{array}{ccc}
0101 & r_4 & 0 \\
0110 & \rightarrow & 0 \\
0111 & & 2 \\
\end{array}
\]
\end{lemma}
\begin{proof}
Recall $\exists \overline{z}$ so that $f(0|1,\ldots,1|\overline{z})=1$. Note that if exists $\overline{z}'$ over $\{0,1\}$ so that $f(0|1,\ldots,1|\overline{z}')=0$ then we have that $f(s(v_1,v'_1),\ldots,s(v_{k+\ell+1},v'_{k+\ell+1}))$ is a generalised Hubie-pol in both $11\ldots 11|00 \ldots 00|\widehat{z}$ and $11\ldots 11|00 \ldots 00|\widehat{z}'$, where we build widehat from overline by doubling each entry where it sits, and we become collapsible. It therefore follows that there must exist distinct $\overline{y}_1$, $\overline{y}_2$, $\overline{z}_1$ and $\overline{z}_2$ (all over $\{0,1\}$) so that $f(0|\overline{y}_1|\overline{z}_1)=0$, $f(0|\overline{y}_2|\overline{z}_2)=0$ but $f(0|\vee(\overline{y}_1,\overline{y}_2)|\overline{z}_1)\neq 0$. By collapsing co-ordinates we get $f'$ so that
\[
\begin{array}{ccc}
0011 & f' & 0 \\
0101 &  \rightarrow & 0 \\
0111 & & \mbox{$0$ or $2$} \\
\end{array}
\]
The result follows by permuting co-ordinates, possibly in new combination through $s$ and the second co-ordinate.
\end{proof}

\subsubsection{The symmetric case: for every $\overline{z} \in \{0,1\}^*$ we have $f(0,\ldots,0|1,\ldots,1|\overline{z})=2$}

\begin{proposition}
\label{prop:symmetric}
Let $f$ be so that neither Class $(i)$ nor Class $(ii)$ is a singleton and so that for every $\overline{z} \in \{0,1\}^*$ we have $f(0,\ldots,0|1,\ldots,1|\overline{z})=2$. Then, either $f$ generates a binary idempotent operation with $01 \mapsto 0$ and $02 \mapsto 2$ or a binary idempotent operation with $01 \mapsto 1$ and $21 \mapsto 2$, or any algebra on $D$ containing $f$ and $s$ is collapsible.
\end{proposition}
\begin{proof}
Let us consider the general form of $f$,
\[
\begin{array}{ccc|ccc|ccccc}
x^1_1 & \cdots & x^k_1 & 1 & \cdots & 1 & w_1^1 & \cdots & w_1^\ell & \ & 1 \\
\vdots & \vdots & \vdots & \vdots & \vdots & \vdots & \vdots & \vdots & \vdots & \mapsto & \vdots \\
x^1_m & \cdots & x^k_m & 1 & \cdots & 1 & w_m^1 & \cdots & w_m^\ell & \ & 1 \\
\\
0 & \cdots & 0 & y^1_1 & \cdots & y^\kappa_1 & z_1^1 & \cdots & z_1^\ell & \ & 0 \\
\vdots & \vdots & \vdots & \vdots & \vdots & \vdots & \vdots & \vdots & \vdots & \mapsto & \vdots \\
0 & \cdots & 0 & y^1_\mu & \cdots & y^\kappa_\mu & z_\mu^1 & \cdots & z_\mu^\ell & \ & 0 \\
\end{array}
\]
where the $x$s, $y$s, $z$s and $w$s are from $\{0,1\}$ and we can assume that each $(x^i_1,\ldots,x^i_m)$ and $(w^i_1,\ldots,w^i_m)$ contain at least one $0$ and $(y^i_1,\ldots,y^i_\mu)$ and $(z^i_1,\ldots,$ $z^i_\mu)$ contains at least one $1$. As in the previous proof we can make an assumption that each $(x^1_i, \cdots, x^k_i , 1, \ldots , 1 , w_i^1 , \ldots , w_i^\ell)$, with $0$ substituted for $2$, still maps under $f$ to $1$. Similarly,  each $(0,\ldots,0,y^1_i, \cdots, y^\kappa_i, z_i^1 , \ldots , z_i^\ell)$, with $1$ substituted for $2$, still maps under $f$ to $0$.

Since for each  $\overline{z} \in \{0,1\}^*$ we have $f(0,\ldots,0|1,\ldots,1|\overline{z})=2$ we can deduce that from the adversaries $(D^{(k+\kappa+\ell-1)i},\{0\}^{M-(k+\kappa+\ell-1)i})$ and  $(D^{(k+\kappa+\ell-1)i},$ $\{1\}^{M-(k+\kappa+\ell-1)i})$, adversaries of the form $(D^{(k+\kappa+\ell)i},\{2\}^{M-(k+\kappa+\ell)i})$, and all co-ordinate permutations of this. 

We now make some case distinctions based on whether $f(0,\ldots,0|2,\ldots,2|$ $2,\ldots,2)=2$ or $0$ and $f(2,\ldots,2|1,\ldots,1|2,\ldots,2)=2$ or $1$ (note that possibly Class $(iii)$ is empty). However, the method for building the full adversary from certain collapsings proceeds very similarly to Cases 2a and 2b from Proposition~\ref{prop:symmetric}. We give an example below as to how, in the case $f(0,\ldots,0|2,\ldots,2|2,\ldots,2)=2$, we mimic Case 2a from Proposition~\ref{prop:asymmetric} to derive a function from this that builds, from adversaries of the form $(D^{r+m-1}, \{0\}^{M-(r+m-1)})$ and $(D^{r+2m-1},\{2\}^{M-(r+m-1)})$, we can build $(D^{r+m},\{0,2\}^{M-m-r})$. For pedagogic reasons we prefer to view some $D$s as $\{0,2\}$s,
\[
\begin{array}{cccc|cccccc}
\{0\} & D & \cdots & D & D & D & \cdots & D & & D \\
D & \{a\} & \cdots & D & D & D & \cdots & D & & D \\
\vdots & \vdots & \cdots & \vdots & \vdots & \vdots & \cdots & \vdots  & & D \\
D & D & \cdots & \{0\} & D & \cdots & D & D & & D \\
\{0\} & \{0\} & \cdots & \{0\}  & \{2\}  & \{0,2\} & \cdots & \{0,2\} & & \{0,2\} \\
\{0\}& \{0\}& \cdots &\{0\}& \{0,2\} & \{2\} & \cdots &\{0,2\} & &\{0,2\}\\
\vdots & \vdots & \vdots & \vdots & \vdots & \vdots & \cdots & \vdots  & \mapsto &\{0,2\} \\
\{0\}&\{0\} & \cdots &\{0\} &\{0,2\}& \{0,2\}& \cdots & \{2\} & & \{0,2\} \\
\{0\} & \{0\}& \cdots & \{0\}& \{2\}  &  \{2\} & \cdots & \{2\}  & & \{2\}  \\
\vdots & \vdots & \vdots & \vdots & \vdots & \vdots & \cdots & \vdots & &  \{2\}  \\
\{0\}&\{0\}& \cdots & \{0\}   &  \{2\} & \cdots & \{2\}  & &  \{2\} \\
\end{array}
\]
\end{proof}

\begin{lemma}
\label{lem:r4-symmetric}
Let $f$ be so that neither Class $(i)$ nor Class $(ii)$ is a singleton and so that for every $\overline{z} \in \{0,1\}^*$ we have $f(0|1,\ldots,1|\overline{z})=2$. Any algebra over $D$ containing $f$ is either collapsible or contains a $4$-ary operations $r^a_4$ and $r^b_4$ with properties
\[
\begin{array}{ccc}
\begin{array}{ccc}
0101 & r^a_4 & 0 \\
0110 & \rightarrow & 0 \\
0111 & & 2 \\
\end{array}
& \ \ \ \ \ \  \mbox{\textbf{and}} \ \ \ \ \ \ &
\begin{array}{ccc}
0101 & r^b_4 & 1 \\
0110 & \rightarrow & 1 \\
0100 & & 2 \\
\end{array}
\end{array}
\]
\end{lemma}
\begin{proof}
The proof proceeds exactly as in Lemma~\ref{lem:r4-asymmetric}.
\end{proof}

An important special case  of the previous lemma, which is satisfied by Chen's $(\{0,1,2\};r,s)$ is as follows.

\vspace{0.2cm}
\noindent \textbf{Zhuk Condition}. $\mathbb{A}$ has idempotent term operations, binary $p$ and ternary operation $r_3$, so that either
\[
\begin{array}{c}
\left(
  \begin{array}{ccc}
    \begin{array}{ccc}
    001 & r_3 & 0 \\
    010 & \rightarrow & 0 \\
    011 & & 2 \\
    \end{array}
  & \ \ \ \ \ \  \mbox{\textbf{and}} \ \ \ \ \ \ &
    \begin{array}{ccc}
    01& p & 0 \\
    02 & \rightarrow & 2 \\
    \end{array}
  \end{array}
\right)
\\
\mbox{\textbf{or}} \\
\left(
\begin{array}{ccc}
   \begin{array}{ccc}
   101 & r_3 & 1 \\
  110  & \rightarrow & 1 \\
  100 & & 2 \\
   \end{array}
& \ \ \ \ \ \  \mbox{\textbf{and}} \ \ \ \ \ \ &
    \begin{array}{ccc}
    01 & p & 1 \\
    21 & \rightarrow & 2 \\
    \end{array}
  \end{array}
\right)
\end{array}
\]

\subsection{About essential relations}

We assume that all relations are defined on the finite set $\{0,1,2\}$.
A relation $\rho$ is called \textit{essential} if
it cannot be represented as a conjunction of relations with smaller arities.
A tuple $(a_{1},a_{2},\ldots,a_{n})$ is called \textit{essential for a relation $\rho$}
if $(a_{1},a_{2},\ldots,a_{n})\notin\rho$ and
for every $i\in \{1,2,\ldots,n\}$ there exists $b\in A$ such that 
$(a_{1},\ldots,a_{i-1},b,a_{i+1},\ldots,a_{n})\in\rho.$ Let us define a relation $\tilde{\rho}$ for every relation $\rho \subseteq D^n$. Put $\sigma_i(x_1,\ldots,x_{i-1},x_{i+1},\ldots,x_{n}) := \exists y \ \rho(x_1,\ldots,x_i,y,x_{i+1},\ldots,x_n)$
and let 
\[ \tilde{\rho}(x_1,\ldots,x_n) := \sigma_1(x_2,x_3,\ldots,x_n) \wedge  \sigma_2(x_1,x_3,\ldots,x_n) \wedge \ldots \wedge  \sigma_1(x_1,x_2,\ldots,x_{n-1}). \]
\begin{lemma}\label{sushnabor}
A relation $\rho$ is essential iff there exists an essential tuple for $\rho$.
\end{lemma}
\begin{proof}
(Forwards.) By contraposition, if $\rho$ is not essential, then $\tilde{\rho}$ is equivalent to $\rho$, and there can not be an essential tuple.

(Backwards.) An essential tuple witnesses that a relation is essential.
\end{proof}
\begin{lemma}
\label{lem:Dmitriy-micro}
Suppose $(2,2,x_3,\ldots,x_n)$ is an essential tuple for $\rho$. Then $\rho$ is not preserved by $s$.
\end{lemma}
\begin{proof}
Since $(2,2,x_3,\ldots,x_n)$ is an essential tuple, $(x_1,c,x_3,\ldots,x_n)$ and $(c,x_2,x_3,$ $\ldots,x_n)$ are in $\rho$ for some $x_1$ and $x_2$. But applying $s$ now gives the contradiction.
\end{proof}
For a tuple $\mathbf{y}$, we denote its $i$th co-ordinate by $\mathbf{y}(i)$. For $n\geq 3$, we define the arity $n+1$ idempotent operation $f^a_n$ as follows
\[
\begin{array}{c}
f^a_n(0,0\ldots,0,0)=0 \\
f^a_n(1,1,\ldots,1,1)=1 \\
f^a_n(1,0,\ldots,0,0)=0 \\
f^a_n(0,1,\ldots,0,0)=0 \\
\vdots \\
f^a_n(0,0,\ldots,1,0)=0 \\
f^a_n(0,0,\ldots,0,1)=0 \\
\mbox{else $2$}
\end{array}
\]
We define $f^b_n$ similarly with $0$ and $1$ swapped. These functions are very similar to partial near-unanimity functions.
\begin{lemma}
\label{lem:Dmitriy}
Suppose $\mathbb{A}$ is a Gap Algebra, that is not $\alpha\beta$-projective, so that $\mathbb{A}$ satisfies the Zhuk Condition. Then either 
\begin{itemize}
\item any relation $\rho \in \mathrm{Inv}(\mathbb{A})$ of arity $h<n+1$ is preserved by $f^a_n$, \textbf{or} 
\item any relation $\rho \in \mathrm{Inv}(\mathbb{A})$ of arity $h<n+1$ is preserved by $f^b_n$.
\end{itemize}
\end{lemma}
\begin{proof}
Suppose \mbox{w.l.o.g.} that the Zhuk Condition is in the first regime and has idempotent term operations, binary $p$ and ternary operation $r_3$, so that
\[
  \begin{array}{ccc}
    \begin{array}{ccc}
    001 & r_3 & 0 \\
    010 & \rightarrow & 0 \\
    011 & & 2 \\
    \end{array}
  & \ \ \ \ \ \  \mbox{\textbf{and}} \ \ \ \ \ \ &
    \begin{array}{ccc}
   01 & p & 0 \\
    02 & \rightarrow & 2 \\
    \end{array}
  \end{array}
\]
We prove this statement for a fixed $n$ by induction on $h$. For $h = 1$ we just need to
check that $f_n:=f^a_n$ preserves the unary relations $\{0,2\}$ and $\{1,2\}$.

Assume that $\rho$ is not preserved by $f_n$, then there exist tuples $\mathbf{y}_1,\ldots,\mathbf{y}_{n+1} \in \rho$ such that $f_n(\mathbf{y}_1,\ldots,\mathbf{y}_{n+1})=\gamma \notin \rho$. We consider a matrix whose columns are $\mathbf{y}_1,\ldots,\mathbf{y}_{n+1}$. Let the rows of this matrix be $\mathbf{x}_1,\ldots,\mathbf{x}_h$.

By the inductive assumption every $\sigma_i$ from the definition of $\widetilde{\rho}$ is preserved by $f_n$, which means that $\widetilde{\rho}$ is preserved by $f_n$, which means that $\gamma \notin \rho$ and $\gamma$ is an essential tuple for $\rho$.

We consider two cases. First, assume that $\gamma$ doesn't contain $2$. Then it follows from the definition that every $\mathbf{x}_i$ contains at most one element that differs from $\gamma(i)$. Since $n+1>h$, there exists $i \in \{1, 2, \ldots , n + 1\}$ such that $\mathbf{y}_i = \gamma$. This contradicts the fact that $\gamma \notin \rho$.

Second, assume that $\gamma$ contains $2$. Then by Lemma~\ref{lem:Dmitriy-micro}, $\gamma$ contains exactly one $2$. \mbox{W.l.o.g.} we assume that $\gamma(1) = 2$. It follows from the definition of $f_n$ that $\mathbf{x}_i$ contains at most one element that differs from $\gamma(i)$ for every $i \in \{2, 3, \ldots , h\}$. Hence, since $n+1>h$, for some $k \in \{1, 2, \ldots , n+ 1\}$ we have $\mathbf{y}_k(i) = \gamma(i)$ for every $i \in \{2, 3, \ldots , h\}$. Since $f_n(\mathbf{x}_1) = 2$, we have one of three subcases. First subcase, $\mathbf{x}_1(j) = 2$ for some $j$. We need one of the properties
\[
\begin{array}{cc|c}
\mathbf{y}_k & \mathbf{y}_j & \gamma \\
\hline
0& 2 & 2 \\
0 & 1 & 0 \\
\end{array}
\mbox{ \ \ \ \ \ \ \ \ \ \ \ \ \ \ \ \ \ \ }
\begin{array}{cc|c}
\mathbf{y}_k & \mathbf{y}_j & \gamma \\
\hline
1 & 2 & 2 \\
0 & 1 & 0 \\
\end{array}
\]
and we can see that the functions from Lemma~\ref{lem:fun} or the definition of the Zhuk Condition suffice, which contradicts our assumptions.

Second subcase, $\mathbf{y}_k(1) = 1, \mathbf{y}_m(1) = 0$ for some $m \in \{1, 2, \ldots , n + 1\}$. We need the property 
\[
\begin{array}{cc|c}
\mathbf{y}_k & \mathbf{y}_m & \gamma \\
\hline
1 & 0 & 2 \\
0 & 1 & 0 \\
\end{array}
\]
and we can check that a function from Lemma~\ref{lem:fun} suffices, which contradicts our assumptions.

Third subcase, $\mathbf{y}_k(1) = 0, \mathbf{y}_m(1) = 1$ and $\mathbf{y}_l(1) = 1$ for $m, l \in \{1, 2, \ldots , n + 1\}\setminus \{k\}$, $m \neq l$. We need the property
\[
\begin{array}{ccc|c}
\mathbf{y}_k & \mathbf{y}_m & \mathbf{y}_l & \gamma \\
\hline
0 & 1 & 1 & 2 \\
0 & 0 & 1 & 0 \\
0 & 1 & 0 & 0 \\
\end{array}
\]
and we can check that the $r_3$ from the Zhuk Condition suffices, which contradicts our assumptions. This completes the proof.
\end{proof}
\begin{corollary}
\label{cor:Dmitriy}
Suppose $\mathbb{A}$ is a Gap Algebra, that is not $\alpha\beta$-projective so that $\mathbb{A}$ satisfies the Zhuk Condition. Then, for every finite subset of $\Delta$ of Inv$(\mathbb{A})$, Pol$(\Delta)$ is collapsible.
\end{corollary}
\begin{proof}
$f^a_n$ is a Hubie-pol in $\{1\}$ and $f^b_n$ is a Hubie-pol in $\{0\}$.
\end{proof}
For $n\geq 2$, we define the arity $n+2$ idempotent operation $\widehat{f}^a_n$ as follows
\[
\begin{array}{c}
f^a_n(0,0,0,\ldots,0,0)=0 \\
f^a_n(1,1,1,\ldots,1,1)=1 \\
f^a_n(1,1,0,\ldots,0,0)=0 \\
f^a_n(1,0,1,\ldots,0,0)=0 \\
\vdots \\
f^a_n(1,0,0,\ldots,1,0)=0 \\
f^a_n(1,0,0,\ldots,0,1)=0 \\
\mbox{else $c$}
\end{array}
\]
We define $\widehat{f}^b_n$ similarly with $0$ and $1$ swapped.
\begin{lemma}
\label{lem:Dmitriy-long}
Suppose $\mathbb{A}$ is a Gap Algebra that is not $\alpha\beta$-projective. Then either 
\begin{itemize}
\item any relation $\rho \in \mathrm{Inv}(\mathbb{A})$ of arity $h<n+2$ is preserved by $\widehat{f}^a_n$, \textbf{or} 
\item any relation $\rho \in \mathrm{Inv}(\mathbb{A})$ of arity $h<n+2$ is preserved by $\widehat{f}^b_n$.
\end{itemize}
\end{lemma}
\begin{proof}
Suppose \mbox{w.l.o.g.} that we are either in the asymmetric case with Class $(i)$ singleton and  there exists $\overline{z} \in \{0,1\}^*$ so that $f(0|1,\ldots,1|\overline{z})=1$ \textbf{or} we are in the symmetric case and we have an idempotent term operation $p$ mapping $01 \mapsto 0$ and $02 \mapsto 2$.

We prove this statement for a fixed $n$ by induction on $h$. For $h = 1$ we just need to
check that $\widehat{f}_n:=\widehat{f}^a_n$ preserves the unary relations $\{0, 2\}$ and $\{1, 2\}$.

Assume that $\rho$ is not preserved by $f_n$, then there exist tuples $\mathbf{y}_1,\ldots,\mathbf{y}_{n+2} \in \rho$ such that $\widehat{f}_n(\mathbf{y}_1,\ldots,\mathbf{y}_{n+2})=\gamma \notin \rho$. We consider a matrix whose columns are $\mathbf{y}_1,\ldots,\mathbf{y}_{n+2}$. Let the rows of this matrix be $\mathbf{x}_1,\ldots,\mathbf{x}_h$.

By the inductive assumption every $\sigma_i$ from the definition of $\widetilde{\rho}$ is preserved by $\widehat{f}_n$, which means that $\widetilde{\rho}$ is preserved by $\widehat{f}_n$, which means that $\gamma \notin \rho$ and $\gamma$ is an essential tuple for $\rho$.

We consider two cases. First, assume that $\gamma$ doesn't contain $2$. Then it follows from the definition that every $\mathbf{x}_i$ contains at most one element that differs from $\gamma(i)$. Since $n+2>h$, there exists $i \in \{1, 2, \ldots , n + 1\}$ such that $\mathbf{y}_i = \gamma$. This contradicts the fact that $\gamma \notin \rho$.

Second, assume that $\gamma$ contains $2$. Then by Lemma~\ref{lem:Dmitriy-micro}, $\gamma$ contains exactly one $2$. \mbox{W.l.o.g.} we assume that $\gamma(1) = 2$. It follows from the definition of $\widehat{f}_n$ that $\mathbf{x}_i$ contains at most one element that differs from $\gamma(i)$ for every $i \in \{2, 3, \ldots , h\}$. Hence, since $n+2>h$, for some $k \in \{2, \ldots , n+2\}$ we have $\mathbf{y}_k(i) = \gamma(i)$ for every $i \in \{2, 3, \ldots , h\}$. Since $\widehat{f}_n(\mathbf{x}_1) = 2$, we have one of four subcases. First subcase, $\mathbf{x}_1(j) = 2$ for some $j$. We need one of the properties
\[
\begin{array}{cc|c}
\mathbf{y}_k & \mathbf{y}_j & \gamma \\
\hline
0 & 2 & 2 \\
0 & 1 & 0 \\
\end{array}
\mbox{ \ \ \ \ \ \ \ \ \ \ \ \ \ \ \ \ \ \ }
\begin{array}{cc|c}
\mathbf{y}_k & \mathbf{y}_j & \gamma \\
\hline
1 & 2 & 2 \\
0 & 1 & 0 \\
\end{array}
\]
and we can see that the functions from Lemma~\ref{lem:fun}, or Proposition~\ref{prop:asymmetric} or Proposition~\ref{prop:symmetric}, suffice which contradicts our assumptions.

Second subcase, $\mathbf{y}_k(1) = 1, \mathbf{y}_m(1) = 0$ for some $m \in \{1, 2, \ldots , n + 1\}$. We need the property 
\[
\begin{array}{cc|c}
\mathbf{y}_k & \mathbf{y}_m & \gamma \\
\hline
1 & 0 & 2 \\
0 & 1 & 0 \\
\end{array}
\]
and we can check that a function from Lemma~\ref{lem:fun} suffices, which contradicts our assumptions.

For Case 3, $\mathbf{y}_k(1) = 0, \mathbf{y}_m(1) = 1$ and $\mathbf{y}_l(1) = 1$ for some $m, l \in \{1, 2, \ldots , n + 1\}\setminus \{k\}$, $m \neq l$ (possibly $1 \in \{m,l\}$). We now split into two subsubcases: either $\mathbf{y}_1(1)=1$ and we need the property
\[
\begin{array}{cccc|c}
\mathbf{y}_1 &  \mathbf{y}_k & \mathbf{y}_m & \mathbf{y}_l & \gamma \\
\hline
1 & 0 & 1 & 1 & 2 \\
1 & 0 & 0 & 1 & 0 \\
1 & 0 & 1 & 0 & 0. \\
\end{array}
\]
Here we can check that $r_4$, from Proposition~\ref{prop:asymmetric} or Proposition~\ref{prop:symmetric}, with co-ordinates $1$ and $2$ permuted, suffices, which contradicts our assumptions. Or we have  $\mathbf{y}_1(1)=0$ and we need the property
\[
\begin{array}{cccc|c}
\mathbf{y}_1 & \mathbf{y}_k & \mathbf{y}_m & \mathbf{y}_l & \gamma \\
\hline
0 & 0 & 1 & 1 & 2 \\
1 & 0 & 0 & 1 & 0 \\
1 & 0 & 1 & 0 & 0. \\
\end{array}
\]
For this $p(x_2,(p_1(x_4,p_1(x_2,x_1))))$ suffices where $p_1$ comes from Lemma~\ref{lem:fun} and $p$ is as before in this proof (\mbox{cf.} Proposition~\ref{prop:asymmetric} and Proposition~\ref{prop:symmetric}).
\[
\begin{array}{cccc|c|c|c}
x_1 & x_2 & x_3 & x_4 & p_1(x_2,x_1) & p_1(x_4,p_1(x_2,x_1)) & p(x_2,(p_1(x_4,p_1(x_2,x_1)))) \\
\hline
0 & 0 & 1 & 1 & 0 & 2 & 2 \\
1 & 0 & 0 & 1 & 1 & 1 & 0\\
1 & 0 & 1 & 0 & 1 & 1 & 0 \\
\end{array}
\]
This completes the proof.
\end{proof}

We are now ready for the main result of this section.

\begin{theorem}
\label{cor:Dmitriy-long}
Suppose $\mathbb{A}$ is a Gap Algebra that is not $\alpha\beta$-projective. Then, for every finite subset of $\Delta$ of Inv$(\mathbb{A})$, Pol$(\Delta)$ is collapsible.
\end{theorem}

\begin{proof}
$\widehat{f}^a_n$ is a Hubie-pol in $\{1\}$ and $\widehat{f}^b_n$ is a Hubie-pol in $\{0\}$.
\end{proof}


\section{Conclusion}
One important application of our abstract investigation of PGP 
yields a nice characterisation in the concrete case of collapsibility,
in particular in the case of a singleton source which we now know can
be equated with preservation under a single polymorphism, namely a
Hubie polymorphism. So far, this is the only known explanation for a
complexity of a (finite signature) QCSP in NP which provokes the following two questions. 
\begin{question}
  For a structure $\mathcal{A}$, is it the case that
  QCSP($\mathcal{A}$) is in NP iff $\mathcal{A}$ admits a Hubie polymorphism?
\end{question}

\begin{question}
  For a structure $\mathcal{A}$, is it the case that
  QCSP($\mathcal{A}$) is in NP iff $\mathcal{A}$ is collapsible?
\end{question}
\noindent Lurking between these questions is the question as to whether collapsibility is always existing from a singleton source (though a better parameter might be obtained from a larger source).


One can further wonder if the parameter
$p$ of collapsibility depends on the size of the structure
$\mathcal{A}$. In particular, this would provide a positive answer to
the following.
\begin{question}
  Given a structure $\mathcal{A}$, can we decide if it is $p$-collapsible
  for some $p$?
\end{question}
Finally, since this paper was drafted, Zhuk has written a new, self-contained, short and elegant proof of Corollary~\ref{cor:in-NP} in the note \cite{Zhuk-PGP-in-NP}.

\section*{Acknowledgment}
The authors would like to thank several anonymous reviewers for their patience,
stamina and very useful suggestions, which have been a great help to
prepare the final version of this paper.





\end{document}